\documentclass[12pt]{article}
\nonstopmode
\RequirePackage[colorlinks,citecolor=blue,urlcolor=blue,linkcolor=blue]{hyperref}
\hypersetup{
pdfauthor = {Runhuan Feng, Alexey Kuznetsov, Fenghao Yang},
pdfkeywords = { exponential functionals, Levy processes, Ornstein-Uhlenbeck process, Mellin transform, Barnes G-function, variable annuity guaranteed benefits},
pdftitle = {Exponential functionals of Levy processes and variable annuity guaranteed benefits},
colorlinks = true,
citecolor=blue,
urlcolor=blue,
linkcolor=blue,
pdfpagemode = UseNone
}
\usepackage[titletoc,title]{appendix}
\usepackage{graphicx,xspace,colortbl}
\usepackage{amsmath,amsthm,amsfonts}
\usepackage{color}
\usepackage{fancybox}
\usepackage{epsfig}
\usepackage{subfigure}
\usepackage{pdfsync}
\usepackage{bbm}
    \def\qed{\hfill$\sqcap\kern-8.0pt\hbox{$\sqcup$}$\\}
    \def\beq{\begin{eqnarray}}
    \def\eeq{\end{eqnarray}}
    \def\beqq{\begin{eqnarray*}}
    \def\eeqq{\end{eqnarray*}}

    \def\re{\textnormal {Re}}
    \def\im{\textnormal {Im}}
    \def\p{{\mathbb P}}
    \def\q{{\mathbb Q}}
    \def\e{{\mathbb E}}
    \def\r{{\mathbb R}}
    \def\c{{\mathbb C}}
    	
    \def\d{{\textnormal d}}
    \def\i{{\textnormal i}}

    \def\mm{{\mathcal M}}

    \def\ee{{\textnormal e}}

\newtheorem{theorem}{Theorem}
\newtheorem{lemma}{Lemma}
\newtheorem{proposition}{Proposition}
\newtheorem{corollary}{Corollary}
\theoremstyle{definition}
\newtheorem{definition}{Definition}

\newtheorem{remark}{Remark}

\title{Exponential functionals of L\'evy processes and variable annuity guaranteed benefits}
\author{
{Runhuan Feng
\footnote{Department of Mathematics, University of Illinois at Urbana-Champaign, USA. Email: rfeng@illinois.edu}} \; ,
{Alexey Kuznetsov
\footnote{Department of Mathematics and Statistics, 
York University, Canada. Email: kuznetsov@mathstat.yorku.ca}}\;, 
and 
{Fenghao Yang
\footnote{Department of Mathematics and Statistics, 
York University, Canada. Email: fenghao@mathstat.yorku.ca}} 
}
\begin{document}
\maketitle

\begin{abstract} Exponential functionals of Brownian motion have been extensively studied in financial and insurance mathematics due to their broad applications, for example, in the pricing of Asian options. The Black-Scholes model is appealing because of mathematical tractability, yet empirical evidence shows that geometric Brownian motion does not adequately capture features of market equity returns. One popular alternative for modeling equity returns consists in replacing the geometric Brownian motion by an exponential of a L\'evy process. In this paper we use this latter model to study variable annuity guaranteed benefits and to compute explicitly the distribution of certain exponential functionals.  
\end{abstract}
{\vskip 0.15cm}
{\vskip 0.15cm}
\section{Introduction}

The study of exponential functionals of Brownian motion has been popularized in finance literature by applications to the pricing of Asian options in financial markets. Asian option is a special type of exotic option contracts whose payoff is contingent upon the average price of underlying asset/equity/commodity over the contract period. In the Black-Scholes model, the evolution of equity value is modeled by a geometric Brownian motion, $\{S_t=S_0 e^{X_t}, t \ge 0\}$ where $X$ is a Brownian motion with drift and $S_0$ is the initial equity value. The continuously monitored Asian call option with a fixed strike price pays off the amount by which the arithmetic average of equity values (from the inception to maturity $T$) exceeds the strike price $K$. In other words, the payoff is
\[\left( \int^T_0 S_t\, \d t -K\right)^+=\left(S_0 J_T -K\right)^+,\] where we have denoted $(x)^+=\max(x,0)$ and 
\begin{equation}\label{def_Jt}
J_t:=\int^t_0 e^{X_s} \, \d s,
\end{equation}
is {\it the exponential functional of the process $X$.}
 Since the no-arbitrage price of an Asian option in the Black-Scholes model is determined by the expected present value of its payoff under a risk-neutral probability measure, the key to the computation of Asian option price is the distribution of the exponential functional $J_T$. There has been a vast amount of work in the literature devoted to the distribution of $J_T$. 
 To name a few, Yor \cite{Yor92} employs the Lamperti transformation relating the geometric Brownian motion and the exponential functional to a Bessel process.  Linetsky \cite{Lin04b} starts with an identity in distribution 
\[J_t \stackrel{d}{=} U_t:= e^{X_t} \int^t_0 e^{-X_s} \, \d s,\] and the fact that the latter is a diffusion process and then applies the eigenfunction expansion technique to determine the distribution of $U_t$.  Vecer \cite{Vec} applies the change of measure to produce a partial differential equation satisfied by the Asian option price. The above list is by no means comprehensive. More applications of exponential functionals of Brownian motion and references can be found in Carmona et al. \cite{CarPetYor} and  Matsumoto and Yor \cite{MatYor1,MatYor2}. 

In several empirical studies (see Cont \cite{Con}, Madan and Seneta \cite{MadSen}, Carr et al. \cite{CarGem}, Kou \cite{Kou})
it was demonstrated that the geometric Brownian motion does not adequately explain many stylized facts of empirical equity returns, such as asymmetric leptokurtic log-returns and volatility smile. One popular solution to this problem is to use L\'evy processes to model log-returns. When working with exponential functionals of L\'evy processes, it is easier to study the distribution of the exponential functional of the form 
\begin{equation}\label{def_I_q}
I_q:=J_{\ee(q)}=\int^{\ee(q)}_0 e^{X_s} \,\d s,
\end{equation}
where $\ee(q)$ is an exponential random variable with mean $1/q$, independent of the process $X$. The first explicit results related to the exponential functional $I_q$ were obtained by  Cai and Kou \cite{Cai_Kou} for hyperexponential L\'evy processes. These results were later extended to processes with jumps of rational transform in \cite{Kuz2012} and to meromorphic L\'evy process in  \cite{KuzHac}. By now the analytical theory behind the exponental functionals $I_q$ is rather well understood, see the papers by Patie and Savov \cite{Patie_Savov,PS2016}.

In this paper we are interested in studying the distribution of a more general exponential functional of the form
\begin{equation}\label{def_J_xt}
J_{x,t}:=x e^{X_t}+\int^t_0 e^{X_s} \,\d s, \;\;\; x\ge 0,
\end{equation}
and of its ``exponential maturity" counterpart 
\begin{equation}\label{def_I_xq}
I_{x,q}:=J_{x,\ee(q)}. 
\end{equation}
Next we will explain the motivation for studying these objects: it comes from certain embedded options in equity-linked insurance, known as {\it variable annuity guaranteed benefits}. 

Equity-linked insurance products allow policyholders to invest their premiums in equity market. In other words, the daily returns on the premium investments are directly linked to a particular equity index, such as S\&P 500, or a particular equity fund of the policyholder's choosing. Upon selection, the premiums are transferred by the insurer to third-party fund managers. To illustrate the mathematical structure, we consider a simplified example. Let $\{F_t, t \ge 0\}$ denote the evolution of a policyholder's investment account and $\{S_t, t \ge 0\}$ denote that of an equity index. Then the equity-linking mechanism dictates that
\begin{align}  \label{equilink}
F_t = F_0 \frac{S_t }{S_0} e^{-mt}, \qquad t \ge 0,
\end{align}
 where $m$ is the rate of account-value-based management and expenses (M\&E) fee per time unit. Among various products, variable annuities are of particular interest as they offer investors a selection of investments often with added guarantees which protect policyholders from severe losses on their investments. These added benefits can often be viewed as the insurance industry's counterparts of option contracts in financial markets. For example, a guaranteed minimum death benefit (GMDB) would guarantee that a policyholder's beneficiary receives the greater of the then-current account value and a guaranteed minimum amount upon the policyholder's death. For example, the guarantee, denoted by $\{G_t, t\ge 0\}$, is for the policyholder to recoup at least his/her initial investment with interest accrued at the risk-free rate, i.e. $G_t=F_0 e^{rt}$, where $r$ is the yield rate per time unit on the insurer's assets backing up the GMDB liability. Denote by $T_x$ the future lifetime of the policyholder, who is currently at age $x$. It is typically assumed in practice that the mortality model is independent of equity returns, i.e. $T_x$ is independent of $\{S_t, t \ge 0\}.$ Therefore, the payoff from the GMDB is given by
\[(G_{T_x}-F_{T_x})^+ ,\] 
which resembles a put option in financial markets. Keep in mind, however, that without any guaranteed benefits the insurer would simply transfer the premiums to third party fund managers. Like other guaranteed benefits, the GMDB is technically an add-on provision to the base contract that provides additional benefits to the policyholder at an additional cost and from which the insurer assumes additional liability. Hence the GMDB is often referred to as a rider. Nonetheless, due to nonforfeiture regulations, the GMDB rider is typically offered on all variable annuity contracts.

While there are many common features of financial derivatives and embedded options in insurance products, a key difference is that financial derivatives are typically short-dated and insurance coverages last for decades. Due to the lack of long-dated options in the market, the risk management of equity-linked insurance is much more sophisticated than the trading of derivatives and plays  a fundamental role to the success of insurance business. In this work, we consider a simplified model that captures the structure of the risk management problem for a variable annuity contract with a plain-vanilla GMDB.

Unlike many exchange-traded financial derivatives which require only an up-front fee, embedded options in equity-linked insurance products are often compensated by a stream of fee incomes. For example, fund managers typically charge a fixed percentage $m$ per time unit per dollar of each policyholder's account and a portion of the fees, say $m_d$, is kicked back to the insurer to compensate for the GMDB rider. Here we consider the present value of the fee income collected continuously up until the time of the policyholder's death,
\[\int^{T \wedge T_x}_0 e^{-rs} m_d F_s\, \d s,\] where $r$ is the yield rate on insurer's bonds backing up the GMDB liability. As in most cases fee incomes exceed the GMDB liability, insurers are interested in the present value of insurer's net liability (gross liability less fee income) 
\[L:=e^{-rT_x}(G_{T_x}-F_{T_x})_+ -\int^{T_x}_0 e^{-rs} m_d F_s \,\d s.\] 
A crucial task of risk management modeling is to quantify and assess the likelihood and severity of positive net liability, which leads to a loss to the insurer. Practitioners typically apply certain risk measures to empirical distributions of net liabilities developed from Monte Carlo simulations. The risk measures would then be used to form the basis of risk management decision making, such as setting up reserves and capitals, to provide a buffer against losses under adverse economic conditions.
The most commonly used risk measures in the North American insurance industry is the conditional tail expectation, 
\[\mathrm{CTE}_p(L)=\e[L|L>\mathrm{VaR}_p(L)],\] where the Value-at-Risk is determined by
\[\mathrm{VaR}_p(L):=\inf\{y: \p[L \le y]\geq p\}. \label{VaR}\]
Since the purpose of risk management is to analyze the severity of positive loss rather than negative loss (profit), we are interested in the risk measures CTE$_p$ and VaR$_p$ for $p>\xi:=\p(L\le 0)$.
 In order to compute the above-mentioned risk measures, we need to compute for $V>\mathrm{VaR}_\xi$ ,
\[ \p(L > V|T_x=t)= \p\left( e^{-rt} F_t+\int^t_0 e^{-rs} m_d F_s \,\d s < F_0- V\right).\] 
It is clear that this rather unique funding mechanism in equity-linked insurance gives rise to a generalized form of exponential functional as defined in \eqref{def_J_xt}. 

While any concern regarding fitting empirical data in the modeling of financial derivatives may carry over to that of equity-linked insurance, there is the additional question of the validity of such models for long-term projection. Nonetheless, the insurance industry has in the past two decades adopted many well-known equity return models from the financial industry, such as geometric Brownian motion, regime-switching geometric Brownian motion, etc. See American Academy of Actuaries publications \cite{AAA}, \cite{AAA09} and \cite{AAApack} for details on a selection of equity return models. Computations of risk measures for variable annuity guaranteed benefits based on exponential functionals of Brownian motion can be found in Feng and Volkmer \cite{FenVol, FenVol2}.  In this paper, we are interested in the exponential L\'evy processes, primarily for two reasons: (i) such models have been shown to explain various stylized facts of empirical data and (ii) they often lead to analytical solutions, not only for pricing problems of exotic options, which are well-studied in finance literature, but also for risk measures of extreme liabilities in equity-linked insurance products, thereby providing fast algorithms for computation needed for capital requirement and other risk management purposes.

This rest of the paper is organized as follows. In section \ref{section_main_results} we study exponential functionals $I_{x,q}$ for general L\'evy processes and derive an integral representation of the Mellin transform of $I_{x,q}$. In section \ref{section_Kou_process} we consider the case of Kou process, and compute the Mellin transform of $I_{x,q}$ explicitly in terms of Meijer G-function, and then identify the density of $I_{x,q}$ (it is also given explicitly in terms of Meijer G-function 
and hypergeometric functions). In section \ref{section_applications} we apply these results to the problem of computing various risk measures for the GMDB and compare the efficiency and accuracy of our semi-analytical approach with the Monte Carlo method.

\section{Main results}\label{section_main_results}

First we introduce the necessary notation and definitions. We consider a L\'evy process $X$, started from zero, and having 
the Laplace exponent $\psi(z):=\ln \e[\exp(z X_1)]$, $z\in \i \r$. The L\'evy-Khintchine formula
tells us that
$$
\psi(z)=\sigma^2 z^2/2+\mu z+\int_{\r} \left(e^{zx}-1-zx {\bf 1}_{\{|x|<1\}} \right) \Pi(\d x), \;\;\; 
z\in \i \r, 
$$
where $\sigma \ge 0$, $\mu \in \r$ and the L\'evy measure $\Pi(\d x)$ satisfies 
$\int_{\r} 1 \wedge x^2 \Pi(\d x)<\infty$. 
We denote by $\ee(q)$ the exponential random variable with mean $1/q$, which is independent of $X$, and we recall our definition of the exponential functional
$$
I_{x,q}:=x e^{X_{\ee(q)}}+\int_0^{\ee(q)} e^{X_s} \d s, \;\;\; x\ge 0. 
$$ 

\begin{remark}
Using time-reversal it is easy to show that $I_{x,q}\stackrel{d}{=} U_{\ee(q)}$, where $U_t$ is the generalized 
Ornstein-Uhlenbeck process 
\begin{equation}\label{OU_process}
U_t=xe^{X_t}+e^{X_t}\int_0^{t}e^{-X_s} \d s. 
\end{equation}
Note that $U_t$ is a strong Markov process started from $x$ with the generator 
$${\mathcal L}^{(U)}f(x)={\mathcal L}^{(X)} \phi(\ln(x))+f'(x),$$ where $\phi(x):=f(e^x)$
and ${\mathcal L}^{(X)}$ is the Markov generator of the L\'evy process $X$. This results follows from \cite[Proposition 2.3]{KPS_2012}. 
\end{remark}

We define the Mellin transform of $I_{x,q}$
\begin{equation}
\label{def_M_xq}
\mm_{x,q}(s)=\e\left[ \left(I_{x,q}\right)^{s-1}\right]. 
\end{equation}
Initially $\mm_{x,q}(s)$ is well defined on the vertical line $\re(s)=1$, later we will extend this function analytically into a certain vertical strip.

Everywhere in this section we will work under the following condition: the measure $\Pi(\d x)$ has exponentially decaying tails. In other words
\begin{equation}\label{tail_condition}
\int_{\r \setminus (-1,1)} e^{\theta |x|} \Pi(\d x)<\infty, \;\;\; 
{\textnormal{ for some }} \theta>0. 
\end{equation}
The above condition implies that the Laplace exponent $\psi(z)$ is analytic in the strip $|\re(z)|<\theta$ and it is convex on the real interval $z \in (-\theta,\theta)$. 

\begin{definition}
For $q>0$ we define 
\begin{align*}
\Phi^+(q)=\sup\{ z>0 \, : \, \psi(z)<q\} \;\;\; 
{\textnormal{ and }} \;\;\; \Phi^-(q)=\inf\{ z<0 \, : \, \psi(z)<q\}. 
\end{align*}
\end{definition}
Note that condition \eqref{tail_condition} implies that for every $q>0$ we have $\Phi^+(q)>0$ and $\Phi^-(q)<0$.

\begin{proposition}\label{prop_finite_Mellin_transform}
 For all $q>0$, $x\ge 0$ and $s\in (0,1+\Phi^+(q))$ we have ${\mathcal M}_{x,q}(s)<\infty$.
 \end{proposition}
 \begin{proof}
 Let us denote $\xi=x \exp(X_{\ee(q)})$ and $\eta=I_{0,q}$, so that $I_{x,q}=\xi+\eta$. Note that 
 $$\e[\xi^{w}]=q/(q-\psi(w))<\infty, \;\;\; w \in (\Phi^-(q), \Phi^+(q))$$ 
 and 
 $\e[\eta^w]<\infty$ for all $w\in (-1,\Phi^+(q))$ (see Rivero \cite[Lemma 2]{Rivero2005}).

 When $0<w<\min(\Phi^+(q),1)$ we use Jensen's inequality and obtain
 $$
 \e[(\xi+\eta)^w]\le \e[\xi^w]+\e[\eta^w]<\infty. 
 $$
 If $\Phi^+(q)>1$, then for $1\le w<\Phi^+(q)$ we use Minkowski inequality
  $$
 \e[(\xi+\eta)^w]^{1/w}\le \e[\xi^w]^{1/w}+\e[\eta^w]^{1/w}<\infty. 
 $$
 Finally, when $-1<w<0$ we use the fact that the function $x \in (0,\infty) \mapsto x^{w}$ is decreasing and obtain
 $$
 \e[(\xi+\eta)^w]<\e[\eta^w]<\infty. 
 $$  
 Thus we have proved that  $\e[(I_{x,q})^w]=\e[(\xi+\eta)^w]<\infty$ for all $w\in (-1,\Phi^+(q)$, which is equivalent to the statement of 
 Proposition \ref{prop_finite_Mellin_transform}.  
 \end{proof}


The following theorem is our main result in this section. 
\begin{theorem}\label{thm_main}
For $q>0$ and $w \in ( \max(-1, \Phi^-(q)) ,0)$
\begin{equation}\label{formula_M_xq_general}
{\mathcal M}_{x,q}(1+w)=q \sin(\pi w) {\mathcal M}_{0,q}(1+w)
\times \left[-\frac{1}{2\i} \int_{c+\i \r} \frac{1}{ z \sin(\pi z) {\mathcal M}_{0,q}(-z)} \times \frac{x^{-z} \d z}{\sin(\pi (w+z))}\right],
\end{equation}
where $c \in (0,-w)$. 
\end{theorem}

Before we prove Theorem \ref{thm_main}, we need to establish several auxiliary results.
\begin{lemma}\label{lemma1}
For $q>0$ the function $F(s)={\mathcal M}_{0,q}(s)/\Gamma(s)$ is analytic and zero-free in the vertical strip 
$\Phi^-(q)<\re(s)<1+\Phi^+(q)$ and it satisfies
\begin{equation}\label{functional_equation_F}
F(s+1)=\frac{1}{q-\psi(s)} F(s), \;\;\; \Phi^-(q)<\re(s)<\Phi^+(q). 
\end{equation} 
\end{lemma}
\begin{proof}
The functional equation follows from  Maulik and Zwart \cite[Lemma 2.1]{Maulik2006}  (see also Carr et al. \cite[Proposition 3.1]{CarPetYor}). 
The fact that $F(s)$ is zero-free follows from the generalized Weierstrass product representation (see Patie and Savov \cite[Theorem 2.1]{Patie_Savov}). 
\end{proof}
 
Let us fix $q>0$, $w \in (\Phi^-(q),0)$ and define a new measure ${\mathbb Q}$  
\begin{equation}\label{def_measure_Q}
\frac{\d {\mathbb Q}}{\d {\mathbb P}} \Big \vert_{{\mathcal F}_t}=e^{w X_t - t \psi(w)}. 
\end{equation}
Under the new measure ${\mathbb Q}$, the process $X$ is a L\'evy process with the Laplace exponent 
$$\psi_{\q}(z)=\psi(z+w)-\psi(w).$$ 
Let us define the exponential functional
\begin{equation}\label{def_J_t}
\hat J_t=\int_0^{t} e^{-X_s} \d s. 
\end{equation}

\begin{lemma}\label{lemma_1_over_M}
 For $w \in (\Phi^-(q),0)$ we denote 
 $\tilde q:=q-\psi(w)$. Then for $0<\re(s)<1+w-\Phi^-(q)$
\begin{equation}\label{eqn_Mellin_J_q}
\e_{\q}\left[ (\hat J_{ e(\tilde q)})^{s-1} \right]=\frac{{\mathcal M}_{0,q}(w)}{\Gamma(w)} \times \frac{\Gamma(s)\Gamma(1+w-s)}{{\mathcal M}_{0,q}(1+w-s)}. 
\end{equation}
\end{lemma}
\begin{proof}
Let us denote $Y_t=-X_t$: under the measure $\q$ this is a L\'evy process with the Laplace exponent 
$\psi_Y(z)=\psi(w-z)-\psi(w)$. Let us also denote $\theta:=w-\Phi^-(q)$ and the function in the right-hand side of 
\eqref{eqn_Mellin_J_q} by $f(s)$. According to Proposition 2 in \cite{Kuz2012}, in order to establish Lemma 
\ref{lemma_1_over_M} we need to check the following three conditions
\begin{itemize}
 \item[(i)] $f(s)$ is analytic and zero-free in the strip $\re(s)\in (0,1+\theta)$,
 \item[(ii)] $f(1)=1$ and $f(s+1)=s f(s)/(\tilde q-\psi_Y(s))$ for all $s\in (0,\theta)$,  
 \item[(iii)] $|f(s)|^{-1}=o(\exp(2 \pi |\im(s)|))$ as $\im(s)\to \infty$, $\re(s)\in (0,1+\theta)$.
\end{itemize}
Condition (i) follows from Lemma \ref{lemma1}. Let us check condition (ii): we use \eqref{functional_equation_F} and check that
\begin{align*}
f(s+1)=\frac{\Gamma(s+1)\Gamma(w-s)}{{\mathcal M}_{0,q}(w-s)}=
\frac{s}{q-\psi(w-s)} \frac{\Gamma(s)\Gamma(w-s+1)}{{\mathcal M}_{0,q}(w-s+1)}
=\frac{s}{\tilde q-\psi_Y(s)} f(s).
\end{align*}
To check condition (iii) we use the well-known asymptotic result
$$|\Gamma(a+\i b)|=\sqrt{2 \pi} \exp(-\pi |b|/2 + (a-1/2)\ln(|b|)+O(1)), \;\;\; b\to \infty,$$
which holds uniformly in $a$ on compact subsets of $\r$, and check that 
$$
\big| 1/f(s) \big |=\Big |\frac{{\mathcal M}_{0,q}(1+w-s)}{\Gamma(s)\Gamma(1+w-s)} \Big |
\le {\mathcal M}_{0,q}(1+w-\re(s)) \times O(e^{3\pi |\im(s)|/2}). 
$$ 
Thus all three conditions are satisfied and we have proved \eqref{eqn_Mellin_J_q}. 
\end{proof}

\vspace{0.2cm}
\noindent
{\bf Proof of Theorem \ref{thm_main}:}
We recall that $I_{x,q}$ has the same distribution as $U_{\ee(q)}=e^{X_{\ee(q)}}(x+\hat J_{\ee(q)})$, where $\hat J_t$ is defined by \eqref{def_J_t}.  
Assume that $q>0$ and $w\in (\max(-1,\Phi^-(q)),0)$, so that $q-\psi(w)>0$. According 
to Proposition \ref{prop_finite_Mellin_transform}, ${\mathcal M}_{x,q}(1+w)<\infty$ and we can write 
\begin{equation}\label{eqn_proof1}
{\mathcal M}_{x,q}(1+w)=\e\left[e^{wX_{\ee(q)}} (x+\hat J_{\ee(q)})^w \right]=
\int_0^{\infty} qe^{-q t} \e\left[e^{wX_t} (x+ \hat J_t)^w \right]\d t.
\end{equation}
Next, define the measure $\q$ as in \eqref{def_measure_Q} and denote $\tilde q=q-\psi(w)$. 
From \eqref{eqn_proof1} we find 
\begin{align}\label{eqn_proof2}
{\mathcal M}_{x,q}(1+w)&=\int_0^{\infty} qe^{-q t} \e\left[e^{wX_t} (x+\hat J_t)^w \right]\d t\\
\nonumber
&=
\int_0^{\infty} qe^{-(q-\psi(w)) t} \e_{\q}\left[(x+\hat J_t)^w \right]\d t=
\frac{q}{\tilde q} \e_{\q} \left[ (x+\hat J_{\ee(\tilde q)})^w \right].
\end{align}

Next, we take $z \in (0,-w)$, use \eqref{eqn_proof2} and compute
\begin{align}\label{eqn_proof3}
\nonumber \int_0^{\infty} x^{z-1} {\mathcal M}_{x,q}(1+w) \d x
&= \frac{q}{\tilde q}  \int_0^{\infty} x^{z-1}\e_{\q} \left[ (x+\hat J_{\ee(\tilde q)})^w \right] \d x \\ \nonumber
&=\frac{q}{\tilde q} \e_{\q} \left[ \int_0^{\infty} x^{z-1} (x+\hat J_{\ee(\tilde q)})^w \d x \right]\\ \nonumber
&=\frac{q}{\tilde q} \e_{\q} \left[ \left(\hat J_{\ee(\tilde q)} \right)^{z+w} \int_0^{\infty} y^{z-1} (y+1)^w \d y \right]\\
&=\frac{q}{\tilde q}  \e_{\q} \left[ \left(\hat J_{\ee(\tilde q)} \right)^{z+w} \right] \times 
\frac{\Gamma(z) \Gamma(-w-z)}{\Gamma(-w)}. \\ \nonumber
&=\frac{q}{\tilde q} \frac{{\mathcal M}_{0,q}(w)}{\Gamma(w)} \times \frac{\Gamma(1+z+w)\Gamma(-z)}{{\mathcal M}_{0,q}(-z)} \times 
\frac{\Gamma(z) \Gamma(-w-z)}{\Gamma(-w)},
\end{align} 
where we used Fubini's theorem in the second step, change of variables $x=J_{\ee(\tilde q)} y$ in the third step, the well-known beta-function integral in the fourth step and Lemma \ref{lemma_1_over_M} in the fifth step. 

Finally, from \eqref{functional_equation_F} we find that
$$
\frac{1}{\tilde q}\frac{{\mathcal M}_{0,q}(w)}{\Gamma(w)}=
\frac{1}{q-\psi(w)}\frac{{\mathcal M}_{0,q}(w)}{\Gamma(w)}=\frac{{\mathcal M}_{0,q}(1+w)}{\Gamma(1+w)}.
$$ 
We also use the reflection formula for the gamma function and rewrite \eqref{eqn_proof3} in the form
$$
\nonumber \int_0^{\infty} x^{z-1} {\mathcal M}_{x,q}(1+w) \d x
=-\frac{ \pi q \sin(\pi w) {\mathcal M}_{0,q}(1+w)}
{ z \sin(\pi z) {\mathcal M}_{0,q}(-z)\sin(\pi (w+z))},
$$
from which formula \eqref{formula_M_xq_general} follows by the inverse Mellin transform. 
\qed

\section{Case study: Kou process}\label{section_Kou_process}

In this section we demonstrate how Theorem \ref{thm_main} can be used to compute explicitly the density of the exponential 
functional $I_{x,q}$ for Kou jump-diffusion process. The latter is defined as follows: 
\begin{align}
X_t=\mu t + \sigma W_t + \sum\limits_{j=1}^{N_t} \xi_i, \label{X}
\end{align}
where $\sigma>0$, $\mu \in \r$,   $N_t$ is a Poisson process with intensity $\lambda$ and $\xi_i$ are i.i.d. random variables having the probability density function
$$
p_{\xi}(x)=p \rho e^{-\rho x} {\mathbf 1}_{\{ x>0\}} + (1-p) \hat \rho e^{\hat \rho x} {\mathbf 1}_{\{x<0\}}, 
$$
for some $p \in (0,1)$ and $\rho, \hat \rho >0$. 
The Laplace exponent is easily seen to be equal to 
$$
\psi(z)=\mu z + \frac{\sigma^2}{2} z^2 + \lambda p \frac{z}{\rho-z}- \lambda (1-p) \frac{z}{\hat{\rho}+z}.   
$$
For $q>0$ the rational function $\psi(z)=q$ has four zeros $\{-\hat\zeta_2,-\hat \zeta_1, \zeta_1, \zeta_2\}$ and two poles 
$\{-\hat \rho,  \rho\}$ which satisfy the interlacing property
$$
-\hat \zeta_2<-\hat \rho<-\hat \zeta_1 < 0 < \zeta_1 <\rho <\zeta_2. 
$$
The Mellin transform ${\mathcal M}_{0,q}(s)$ was computed in Cai and Kou \cite{Cai_Kou} (see also \cite{Kuz2012}) and is given by
\begin{equation}\label{eqn_M_0q_Kou_process}
{\mathcal M}_{0,q}(s)=A^{1-s} \Gamma(s) \frac{\mathcal G(s)}{\mathcal G(1)},
\end{equation} 
where $A=\sigma^2/2$ and 
$$
{\mathcal G}(s):=
\Gamma\Big[ \begin{array}{c}
1+\zeta_1-s, \; 1+\zeta_2-s, \; \hat \rho+s \\
1+\rho-s, \;\hat \zeta_1+s, \;\hat \zeta_2+s 
\end{array} \Big].
$$
In the above formula (and everywhere else in this paper) we use the notation 
\begin{equation}\label{gamma_product_notation}
\Gamma\Big[ \begin{array}{c}
a_1, \dots, a_p \\
b_1, \dots, b_q
\end{array} \Big] := \frac{\prod_{i=1}^p \Gamma(a_i)}{ \prod_{j=1}^q \Gamma(b_j)}.
\end{equation}

Our first main result in this section is an explicit expression for the Mellin transform ${\mathcal M}_{x,q}(s)$. 
\begin{proposition}\label{prop_Mellin_transform_Kou} For $0 \vee (1-\hat \zeta_1) < \re(s)< 1$ 
\begin{align}\label{eqn_formula_Mxq}
{\mathcal M}_{x,q}(s)&=q A^{-s} 
\Gamma\Big[ \begin{array}{c}
1+\zeta_1-s, \; 1+\zeta_2-s, \; \hat \rho+s \\
1-s, \; 1+\rho-s, \;\hat \zeta_1+s, \;\hat \zeta_2+s 
\end{array} \Big]
G_{4,5}^{3,3}
\Big( \begin{array}{c}
1-s, 1, -\rho, \hat \rho \\
1-s, \hat \zeta_1, \hat \zeta_2, -\zeta_1,-\zeta_2 
\end{array} \Big \vert \frac{1}{Ax} \Big),
\end{align} 
where $G$ is the Meijer G-function defined in \eqref{def_Meijer_G} of Appendix \ref{AppendixA}.
\end{proposition}
\begin{proof}
Formula \eqref{eqn_M_0q_Kou_process} and Theorem \ref{thm_main} tells us that for 
$- (1 \wedge \hat \zeta_1) <w<-c<0$ we have 
\begin{align*}
{\mathcal M}_{x,q}(1+w)&=
q \sin(\pi w) A^{-w} 
\Gamma\Big[ \begin{array}{c}
1+w, \zeta_1-w, \; \zeta_2-w, \; \hat \rho+1+w \\
\rho-w, \;\hat \zeta_1+1+w, \;\hat \zeta_2+1+w 
\end{array} \Big]\\
&\times 
\frac{-1}{2i} 
\int_{c+\i \r} 
\Gamma\Big[ \begin{array}{c}
1+\rho+z, \;\hat \zeta_1-z, \;\hat \zeta_2-z \\
-z, 1+\zeta_1+z, \; 1+\zeta_2+z, \; \hat \rho-z 
\end{array} \Big]
\frac{A^{-1-z} x^{-z} \d z}{z \sin(\pi z) \sin(\pi (w+z))}.
\end{align*}
Using the reflection formula for the Gamma function we rewrite the above equation in the form
\begin{align*}
{\mathcal M}_{x,q}(1+w)&=
q  A^{-1-w} 
\Gamma\Big[ \begin{array}{c}
\zeta_1-w, \; \zeta_2-w, \; \hat \rho+1+w \\
-w, \;  \rho-w, \;\hat \zeta_1+1+w, \;\hat \zeta_2+1+w 
\end{array} \Big]\\
&\times 
\frac{1}{2\pi i} 
\int_{c+\i \r} 
\Gamma\Big[ \begin{array}{c}
1+w+z, \; z, \; 1+\rho+z, \; -w-z , \;\hat \zeta_1-z, \;\hat \zeta_2-z \\
\hat \rho-z, \; 1+\zeta_1+z, \; 1+\zeta_2+z
\end{array} \Big]
(Ax)^{-z} \d z.  
\end{align*}
Applying formula \eqref{def_Meijer_G} we conclude that for all 
$- (1 \wedge \hat \zeta_1) <w<0$
\begin{align}\label{proof_eqn_Mxq_1}
{\mathcal M}_{x,q}(1+w)&=
q  A^{-1-w} 
\Gamma\Big[ \begin{array}{c}
\zeta_1-w, \; \zeta_2-w, \; \hat \rho+1+w \\
-w, \;  \rho-w, \;\hat \zeta_1+1+w, \;\hat \zeta_2+1+w 
\end{array} \Big]\\
\nonumber
&\times 
G_{5,4}^{3,3}
\Big( \begin{array}{c}
1+w, 1-\hat \zeta_1, 1-\hat \zeta_2, 1+\zeta_1,1+\zeta_2 \\
1+w, 0, 1+\rho, 1-\hat \rho \\ 
\end{array} \Big \vert Ax \Big).  
\end{align}
Note that both conditions \eqref{condition_A} and \eqref{condition_B} are satisfied, since in our case we have
\begin{align*}
a&=\max(1+w,1-\hat \zeta_1,1-\hat \zeta_2)=\max(1+w,1-\hat \zeta_1) \in (0,1),\\
b&=\min(0,1+w,1+\rho)=0.
\end{align*}
and $c \in (-b,1-a)$. The desired result \eqref{eqn_formula_Mxq} is obtained from \eqref{proof_eqn_Mxq_1}  by changing the variable $w=s-1$ and applying formula \eqref{Meijer_G_x_to_1/x}. 
\end{proof}

For the rest of this section we will work under the following

\vspace{0.25cm}
\noindent
{\bf Assumption 1:}  $\zeta_2-\zeta_1 \notin {\mathbb N}$ and $\hat \zeta_2-\hat \zeta_1 \notin {\mathbb N}$.

\vspace{0.25cm}
\begin{definition}
We define the function $f_{x,q}(y)$ as follows:  for $y>x$ 
\begin{align}\label{density function right side of x}
f_{x,q}(y)&:=\Big \{\frac{qx^{\zeta_1}+\zeta_1{\mathcal M}_{x,q}(\zeta_1)}{\psi'(\zeta_1)}
y^{-1-\zeta_1} 
{}_3F_3 \Big( \begin{array}{c}
1+\zeta_1, 1+\zeta_1-\rho, 1+\zeta_1+\hat \rho \\
 1+\zeta_1-\zeta_2, 1+\zeta_1+\hat \zeta_1, 1+\zeta_1+\hat \zeta_2 
\end{array} \Big \vert -\frac{1}{Ay} \Big) \Big \}\\
\nonumber
&\;+ \Big \{{\textnormal{the same expression with $\zeta_1$ and  $\zeta_2$ interchanged}}\Big \},
\end{align}
and for $0<y<x$ 
\begin{align}\label{density function left side of x}
f_{x,q}(y)&:=
\Big \{q (Ax)^{-\hat \zeta_1}\frac{\sin(\pi(\hat \rho-\hat \zeta_1))}{\sin(\pi(\hat \zeta_2-\hat \zeta_1))}
{}_3\Phi_3 \Big( \begin{array}{c}
\hat \zeta_1, 1+\hat \zeta_1+\rho, 1+\hat \zeta_1-\hat \rho \\
1+\hat \zeta_1-\hat \zeta_2, 1+\hat \zeta_1+\zeta_1, 1+\hat \zeta_1+\zeta_2
\end{array} \Big \vert \frac{1}{Ax} \Big) \\
\nonumber
&\qquad \qquad\qquad \qquad\qquad \;\;\; \times
G_{3,4}^{3,1}
\Big( \begin{array}{c}
 1-\hat \rho, 1, 1+\rho \\
 1+\zeta_1, 1+\zeta_2, 1-\hat \zeta_1, 1-\hat \zeta_2
\end{array} \Big \vert \frac{1}{Ay} \Big)\Big\}\\ \nonumber
&\;+ \Big\{{\textnormal{the same expression with $\hat \zeta_1$ and  $\hat \zeta_2$ interchanged}}\Big\}.
\end{align}
In the above formula $\Phi$ denotes the regularized hypergeometric function, as defined in \eqref{def_pPhir} of Appendix \ref{AppendixA}.
\end{definition}

\begin{theorem}\label{thm_density_function} 
The probability density function of $I_{x,q}$ is $f_{x,q}(y)$.
\end{theorem} 
\begin{proof}
Applying formula \eqref{Meijer_G_asymptotics}, we check that for any $\epsilon>0$ small enough 
\begin{align}
\label{f_x_q_asymptotics_zero}
f_{x,q}(y)&=O(y^{\hat \rho-\epsilon}), \; {\textnormal{ as }} \;\;y\rightarrow 0,  \\
\label{f_x_q_asymptotics_infinity}
f_{x,q}(y)&=O(y^{-1-\zeta_1}),\; {\textnormal{ as }} \;\;y \rightarrow +\infty, 
\end{align}
so the function $y^{s-1} f_{x,q}(y)$ is integrable for $0 \vee (1-\hat \zeta_1) < \re(s)< 1$. For $s$ in this strip
we  define 
$$
I_1(s):=\int_0^{x} f_{x,q}(y) y^{s-1} \d y, 
\qquad \qquad 
I_2(s)=\int_x^{\infty} f_{x,q}(y) y^{s-1} \d y, 
$$
and now our goal is to check that $I_1(s)+I_2(s)={\mathcal M}_{x,q}(s)$ (where the right-hand side is given by 
\eqref{eqn_formula_Mxq}). 

First we use formula \eqref{Meijer_G_definite_integral} and obtain
\begin{align*}
I_1(s)&=\Big \{q A^{-\hat \zeta_1} x^{s-\hat \zeta_1}\frac{\sin(\pi(\hat \rho-\hat \zeta_1))}{\sin(\pi(\hat \zeta_2-\hat \zeta_1))}
{}_3\Phi_3 \Big( \begin{array}{c}
\hat \zeta_1, 1+\hat \zeta_1+\rho, 1+\hat \zeta_1-\hat \rho \\
1+\hat \zeta_1-\hat \zeta_2, 1+\hat \zeta_1+\zeta_1, 1+\hat \zeta_1+\zeta_2
\end{array} \Big \vert \frac{1}{Ax} \Big) 
 \\
\nonumber
&\qquad \qquad\qquad \qquad\qquad \;\;\; \times
G_{4,5}^{4,1}
\Big( \begin{array}{c}
 1-\hat \rho, 1, 1+\rho, s+1 \\
 s, 1+\zeta_1, 1+\zeta_2, 1-\hat \zeta_1, 1-\hat \zeta_2
\end{array} \Big \vert \frac{1}{Ax} \Big)\Big\}\\
\nonumber &+ \Big\{{\textnormal{the same expression with $\hat \zeta_1$ and  $\hat \zeta_2$ interchanged}}\Big\}. 
\end{align*}
Similarly, using formula \eqref{pFq_definite_integral} we find
\begin{align*}
I_2(s)&=\Big\{\frac{qx^{\zeta_1}+\zeta_1 {\mathcal M}_{x,q}(\zeta_1)}{\psi'(\zeta_1)}
\frac{x^{s-1-\zeta_1}}{1+\zeta_1-s} 
{}_4F_4 \Big( \begin{array}{c}
1+\zeta_1-s, 1+\zeta_1, 1+\zeta_1-\rho, 1+\zeta_1+\hat \rho \\
2+\zeta_1-s, 1+\zeta_1-\zeta_2, 1+\zeta_1+\hat \zeta_1, 1+\zeta_1+\hat \zeta_2 
\end{array} \Big \vert -\frac{1}{Ax} \Big)\Big\} \\
\nonumber
&+ \Big\{{\textnormal{the same expression with $\zeta_1$ and  $\zeta_2$ interchanged}}\Big\}. 
\end{align*}

Let us outline the plan for proving the identity 
\begin{equation}\label{identity_I1_I2_M}
I_1(s)+I_2(s)-{\mathcal M}_{x,q}(s)=0, \;\;\; {\textnormal{ for all $s$ in the strip $0 \vee (1-\hat \zeta_1) < \re(s)< 1$}}.
\end{equation}
First we use formula \eqref{Meijer_G_in_terms_of_pFq} and express all Meijer G-functions appearing in \eqref{identity_I1_I2_M} in terms of hypergeometric functions. 
This would give us an expression involving products of two hypergeometric functions. 
After simplifying this expression we would obtain the following identity
\begin{align}\label{4F4_identity}
\sum\limits_{i=1}^{5} 
\frac{(a_i-\rho)(a_i+\hat \rho)}{ \prod\limits_{\stackrel{1\le j \le 5}{j \neq i}} (a_i-a_j)}  &\times 
{}_{4}F_{4} \Big (
\begin{matrix}
1+a_i-\rho, 1+a_i+\hat \rho, 1+a_i, 1+a_i-s \\
1+a_i-a_1, \dots,*,\dots, 1+a_i-a_{5}
\end{matrix} \Big \vert 
-\frac{1}{Ax} \Big)
\\ \nonumber
&\times 
{}_{4}F_{4} \Big (
\begin{matrix}
 1+\rho-a_i,1-\hat \rho-a_i, -a_i, s-a_i \\
1+a_1-a_i, \dots, *, \dots, 1+a_{5}-a_i
\end{matrix} \Big \vert 
\frac{1}{Ax} \Big)=0, \;\;\; x \in {\mathbb R}\setminus \{0\}, 
\end{align}
where $[a_1, a_2, a_3, a_4, a_5]=[\zeta_1, \zeta_2, -\hat \zeta_1, -\hat \zeta_2, s-1]$
and the asterisk means that the term $1+a_i-a_i$ is omitted. The identity \eqref{4F4_identity} is known to be true: it is a special case of Theorem 1 in Feng et al. \cite{FKY2015}.  

The above steps of the proof, while conceptually simple, require very long computations. Therefore, we omit here all these details and we present them in Appendix \ref{AppendixB}.
\end{proof}

\begin{remark}
The algebraic manipulations needed in the last step of the proof of Theorem \ref{thm_density_function} (where we establish identity \eqref{identity_I1_I2_M}) are rather tedious (as can be seen in the Appendix \ref{AppendixB}). At the same time, it is  easy to confirm the validity of this identity by a numerical experiment: 
one simply needs to compute Meijer G-functions via \eqref{Meijer_G_in_terms_of_pFq} and the hypergeometric functions via series expansion \eqref{def_pFr}, and check that \eqref{identity_I1_I2_M} holds true with arbitrary choices of parameters. 
\end{remark}

In the next result we compute the distribution function of $I_{x,q}$. 

\begin{corollary}\label{cor_CDF}
For $y\ge x$
\begin{align}\label{CDF1}
\p(I_{x,q}>y)&=\Big \{ \frac{qx^{\zeta_1}+\zeta_1 {\mathcal M}_{x,q}(\zeta_1)}{\zeta_1 \psi'(\zeta_1)}
y^{-\zeta_1}
{}_3F_3 \Big( \begin{array}{c}
 1+\zeta_1-{ \rho}, 1+\zeta_1+\hat { \rho}, \zeta_1 \\
 1+\zeta_1-\zeta_2, 1+\zeta_1+\hat \zeta_1, 1+\zeta_1+\hat \zeta_2 
\end{array} \Big \vert -\frac{1}{Ay} \Big)\Big \}
\\
\nonumber
&+ \Big\{{\textnormal{the same expression with $\zeta_1$ and  $\zeta_2$ interchanged}}\Big\}
\end{align}
and for $0<y<x$ 
\begin{align}\label{CDF2}
\p(I_{x,q}<y)&=\Big \{\frac{q}{A} (Ax)^{-\hat \zeta_1}\frac{\sin(\pi(\hat { \rho}-\hat \zeta_1))}{\sin(\pi(\hat \zeta_2-\hat \zeta_1))}
{}_3\Phi_3 \Big( \begin{array}{c}
\hat \zeta_1, 1+\hat \zeta_1+{ \rho}, 1+\hat \zeta_1-\hat { \rho} \\
1+\hat \zeta_1-\hat \zeta_2, 1+\hat \zeta_1+\zeta_1, 1+\hat \zeta_1+\zeta_2
\end{array} \Big \vert \frac{1}{Ax} \Big) \\
\nonumber
&\qquad \qquad\qquad \qquad\qquad \;\;\; \times
G^{3,1}_{3,4}
\Big( \begin{array}{c}
 -\hat \rho, \rho, 1 \\
 \zeta_1, \zeta_2, -\hat \zeta_1, -\hat \zeta_2
\end{array} \Big \vert \frac{1}{Ay} \Big)\Big\}
\\
\nonumber
&+ \Big\{{\textnormal{the same expression with $\hat \zeta_1$ and  $\hat \zeta_2$ interchanged}}\Big\}.
\end{align}

\end{corollary} 
\begin{proof}
Formula \eqref{CDF1} can be easily obtained from \eqref{density function right side of x} and \eqref{pFq_definite_integral}. Similarly, formula \eqref{CDF2} follows from \eqref{density function left side of x}, \eqref{Meijer_G_definite_integral}, 
\eqref{Meijer_G_reduction} and \eqref{Meijer_G_times_x_power_c}. 
\end{proof}

\begin{corollary}\label{cor_ECT}
Assume that $\zeta_1>1$. Then for $y\ge  x$
\begin{align}
\e[I_{x,q} \mathbf{1}_{\{I_{x,q}>y\}}]&=
\Big\{\frac{qx^{\zeta_1}+\zeta_1 {\mathcal M}_{x,q}(\zeta_1)}{\psi'(\zeta_1)(\zeta_1-1)}
y^{1-\zeta_1}
{}_4F_4 \Big( \begin{array}{c}
1+\zeta_1, 1+\zeta_1-{ \rho}, 1+\zeta_1+\hat { \rho}, \zeta_1-1 \\
 1+\zeta_1-\zeta_2, 1+\zeta_1+\hat \zeta_1, 1+\zeta_1+\hat \zeta_2, \zeta_1 
\end{array} \Big \vert -\frac{1}{Ay} \Big) \Big\}\\
\nonumber
&+ \Big\{{\textnormal{the same expression with $\zeta_1$ and  $\zeta_2$ interchanged}}\Big\},
\end{align}
and for $0<y<x$ 
\begin{align}
\e[I_{x,q} \mathbf{1}_{\{I_{x,q}<y\}}]&=
\Big\{\frac{qy^2 \sin(\pi(\hat { \rho}-\hat \zeta_1))}{(Ax)^{\hat \zeta_1}\sin(\pi(\hat \zeta_2-\hat \zeta_1))}
{}_3\Phi_3 \Big( \begin{array}{c}
\hat \zeta_1, 1+\hat \zeta_1+{ \rho}, 1+\hat \zeta_1-\hat { \rho} \\
1+\hat \zeta_1-\hat \zeta_2, 1+\hat \zeta_1+\zeta_1, 1+\hat \zeta_1+\zeta_2
\end{array} \Big \vert \frac{1}{Ax} \Big) \\
\nonumber
&\qquad \qquad\qquad \qquad\qquad \;\;\; \times
G^{4,1}_{4,5}
\Big( \begin{array}{c}
 1-\hat \rho, 1, 1+\rho, 3 \\
 2, 1+\zeta_1, 1+\zeta_2, 1-\hat \zeta_1, 1-\hat \zeta_2
\end{array} \Big \vert \frac{1}{Ay} \Big)\Big\}\\ 
\nonumber
&+ \Big\{{\textnormal{the same expression with $\hat \zeta_1$ and  $\hat \zeta_2$ interchanged}}\Big\}.
\end{align}
\end{corollary}
\begin{proof}
Same steps as in the proof of Corollary \ref{cor_CDF}. 
\end{proof}

\begin{remark} \label{remBM}
The jump-diffusion process \eqref{X} includes the Brownian motion with drift $X_t=\mu t+\sigma W_t$ as a special case when $\lambda=0$. In this case the expressions in Corollary \ref{cor_CDF} can be simplified. Let us denote
\[\nu:=\frac{2\mu}{\sigma^2}, \;\;\; \eta:=\frac{\sqrt{8q/\sigma^2+\nu^2}}2,  \;\;\; \kappa:=\frac{1-\nu}2.\]
 Then for $y\ge x$,
\begin{align*}
\p(I_{x,q}>y)=q \frac{\Gamma(\eta-\kappa+1/2) x^\kappa y^{1-\kappa}}{\Gamma(1+2\eta)(\eta+\kappa-1/2)}
e^{(1/x-1/y)/\sigma^2}  W_{\kappa,\eta}\Big(\frac{2}{\sigma^2 x}\Big) M_{\kappa-1,\eta}\Big( \frac2{\sigma^2 y}\Big),
\end{align*} 
while for $0<y< x$,
\begin{align*}
\p(I_{x,q}<y)=q \frac{\Gamma(\eta-\kappa+1/2) x^\kappa y^{1-\kappa}}{\Gamma(1+2\eta)}e^{(1/x-1/y)/\sigma^2} M_{\kappa,\eta}\Big(\frac{2}{\sigma^2 x}\Big) W_{\kappa-1,\eta}\Big( \frac2{\sigma^2 y}\Big).
\end{align*}
Here $M$ and $W$ denote the Whittaker functions, whose definitions and basic properties can be found in Olver et al. \cite{Olver}. 
 Similarly, expressions in Corollary \ref{cor_ECT} can be simplified: for $y\ge x$,
\begin{align*}
\e[I_{x,q} \mathbbm{1}_{\left\{ I_{x,q}>y\right\}}]= \frac{q\Gamma(\eta-\kappa+1/2) x^\kappa y^{2-\kappa}}{\Gamma(1+2\eta)(\eta+\kappa-1/2)}e^{(1/x-1/y)/\sigma^2}W_{\kappa,\eta}\Big(\frac{2}{\sigma^2 x}\Big) \Bigg[ \frac{M_{\kappa-2,\eta}\Big( \frac2{\sigma^2 y}\Big)}{\eta+\kappa-3/2} +M_{\kappa-1,\eta}\Big( \frac2{\sigma^2 y}\Big)\Bigg],
\end{align*} and for $0<y<x$,
\begin{align*}
\e[I_{x,q} \mathbbm{1}_{\left\{ I_{x,q}<y\right\}}]= \frac{q\Gamma(\eta-\kappa+1/2) x^\kappa y^{2-\kappa}}{\Gamma(1+2\eta)}e^{(1/x-1/y)/\sigma^2} M_{\kappa,\eta}\Big(\frac{2}{\sigma^2 x}\Big) \left[W_{\kappa-1,\eta}\Big( \frac2{\sigma^2 y}\Big)- W_{\kappa-2,\eta}\Big( \frac2{\sigma^2 y}\Big) \right].
\end{align*}  These expressions were obtained in Feng and Volkmer \cite[Proposition 3.4]{FenVol2} using spectral methods.
\end{remark}

\section{Applications}\label{section_applications}

As we have discussed in the introduction, exponential functionals arise naturally in the analysis of insurer's liabilities to variable annuity guaranteed benefits, due to the continual collection of management fees as a fixed percentage of policyholders' account value. In this section we apply our theoretical results obtained earlier and we compute various risk measures for the 
guaranteed minimum death benefit (GMDB), which is one of the most common types of investment guarantees in the market.

Assume that the equity index $\{S_t, t \ge 0\}$ is modeled by an exponential L\'evy process
\[S_t:=S_0 e^{X_t},\qquad t \ge 0,\] where $X$ is the Kou process, as defined in \eqref{X}. Assume, also, that the
 policyholder's investment account is driven by the equity-linking mechanism as in \eqref{equilink}. 
Recall that the GMDB net liability from an insurer's viewpoint is given by
\begin{align} \label{Lform} 
L:=e^{-rT_x}(F_0 e^{rT_x} -F_{T_x})_+-\int^{ T_x}_0 e^{-rs} m_d F_s \,\d s.
\end{align} 
Due  to the independence of mortality risk and equity risk, we obtain an expression of $\p(L>V)$ for $V \ge \mathrm{VaR}_\xi>0,$
\begin{align} 
\p(L > V)= \int^\infty_0 P(t,K) f(t) \, \d t,
\end{align} where $f$ is the probability density function of $T_x$,  $K:=(F_0-V)/(m_d F_0)$ and
\[ P(t,K):= \p \left( x e^{X^\ast_t}+\int^t_0 e^{X^\ast_s } \d s <K  \right)\] with $x=1/m_d$. The underlying L\'evy process 
$X^*$ is the same as the process $X$ in \eqref{X}, but with with $\mu$ replaced by 
\[\mu^\ast:=\mu-r-m.\]  The Laplace transform of $P$ with respect to $t$ is given by
\[ \tilde{P}(q,K):=\int^\infty_0 e^{-qt} P(t, K) \d t=\frac1q\p( I_{x,q}<K).\]
Similarly, we can show that
\begin{align} \mathrm{CTE}_p(L)=F_0 -\frac{m_d F_0}{1-p}   \int^\infty_0 Z(t,K) f(t) \d t,\label{CTE}\end{align} where
\[Z(t,K):= \e\left[\left(x e^{X^\ast_t}+\int^t_0 e^{X^\ast_s } \d s \right)\mathbbm{1}_{\left\{x e^{X^\ast_t}+\int^t_0 e^{X^\ast_s } \d s <K \right\}}\right].\] Its Laplace transform with respect to $t$ is given by
\[\tilde{Z}(q,K):=\int^\infty_0 e^{-qt} Z(t, K) \d t=\frac1q \e\left [I_{x,q} \mathbbm{1}_{\{I_{x,q}<K\}}\right]. \]

A common model for human mortality in the literature is the so-called Gompertz-Makeham law of mortality, which assumes that the death rate $\mu_x$ is the sum of a constant $A$ (to account for death due to accidents) and a component $B c^x$ (to account for aging): 
\[\mu_x =A+B c^x, \qquad A>0, B>0, c>1.\]
Its probability density function $f$ is given by
\begin{align}
f(t)=(A+Bc^{x+t})\exp \left\{-At-\frac{Bc^x (c^t-1)}{\ln c}\right\}. \label{morden}
\end{align}
As shown in Feng and Jing \cite{FJ}, we can always use a decomposition of a Hankel matrix to approximate $f$ by a combination of exponential functions with complex components and complex weights,
\[f(t)\approx \sum^M_{i=1} w_i e^{- s_i t}, \qquad \Re(s_i)>0.\]
There are many known methods in the literature for such approximations, most of which utilizes only real components and real weights. However, the Hankel matrix method has the advantage of using relatively small number of terms. Then, for large enough $M$,
\begin{align} \p(L>V)\approx \sum^M_{i=1} w_i \tilde{P}(s_i, K).\label{appx}\end{align}
Similarly, we can approximate the CTE risk measure by
\begin{align} \mathrm{CTE}_p(L)\approx F_0 -\frac{m_d F_0}{1-p} \sum^M_{i=1} w_i \tilde{Z}(s_i, K).\label{appxCTE}\end{align}

Let us illustrate the application to GMDB with a numerical example.

\begin{figure}[t]
	\centering
	\subfigure[Bases]{\includegraphics[width=6cm]{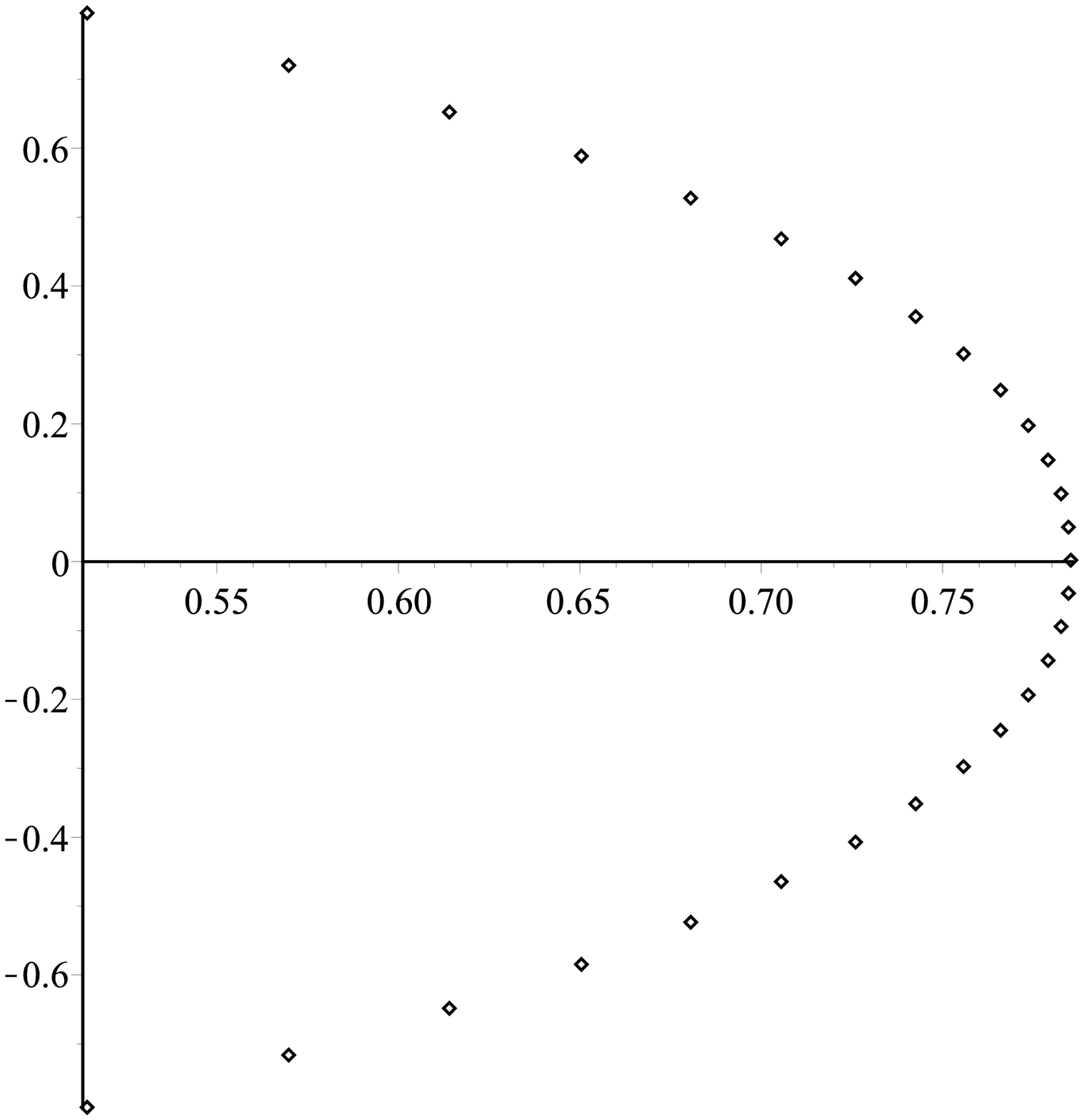}}\qquad
	\subfigure[Weights]{\includegraphics[width=6cm]{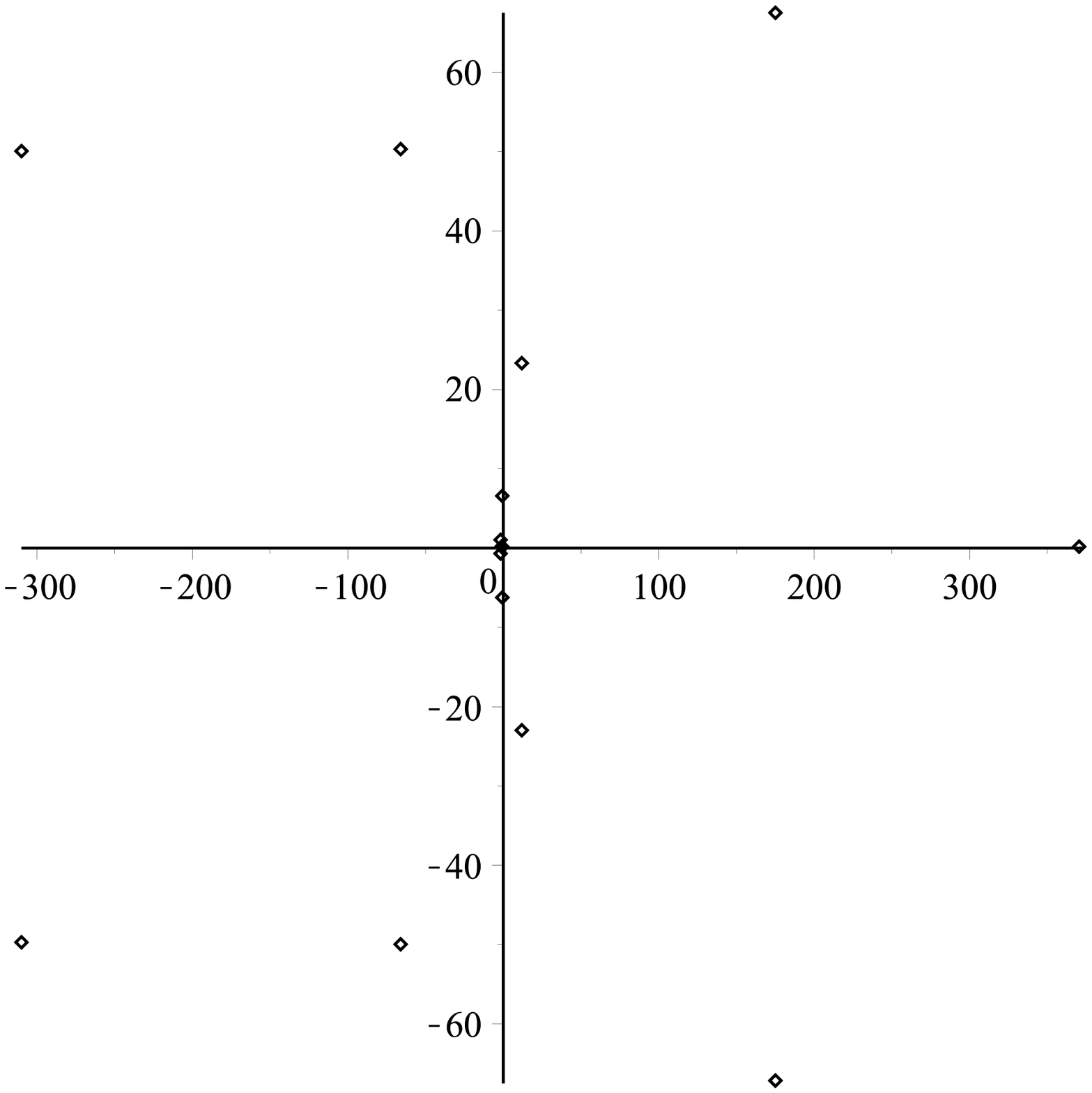}}
	\caption{Approximating exponential sum}
	\label{fig:weight}
\end{figure}

\begin{figure}[t]
	\centering
	\subfigure[Mortality density]{\includegraphics[width=6cm]{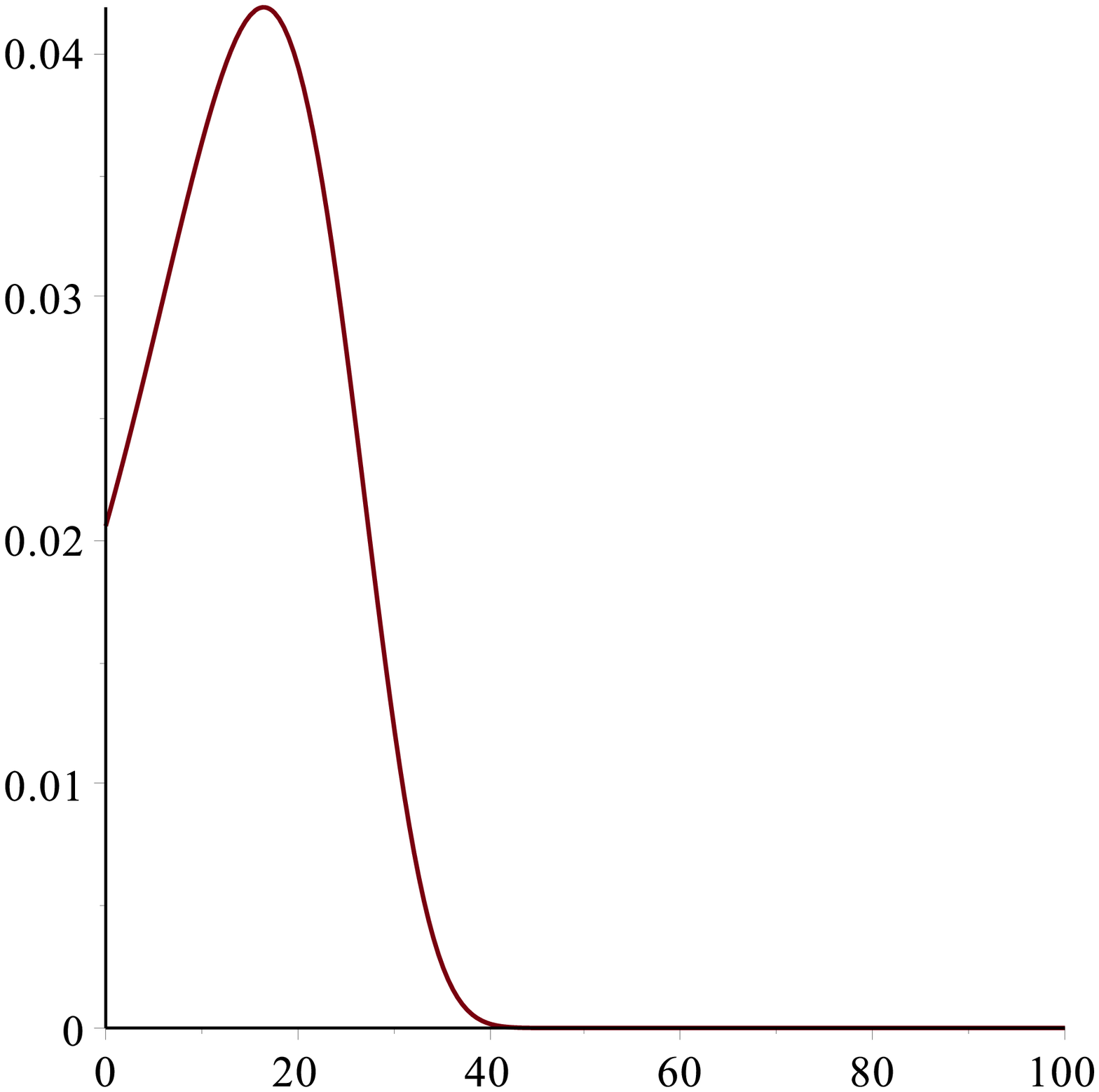}}\qquad
	\subfigure[Approximation error]{\includegraphics[width=6cm]{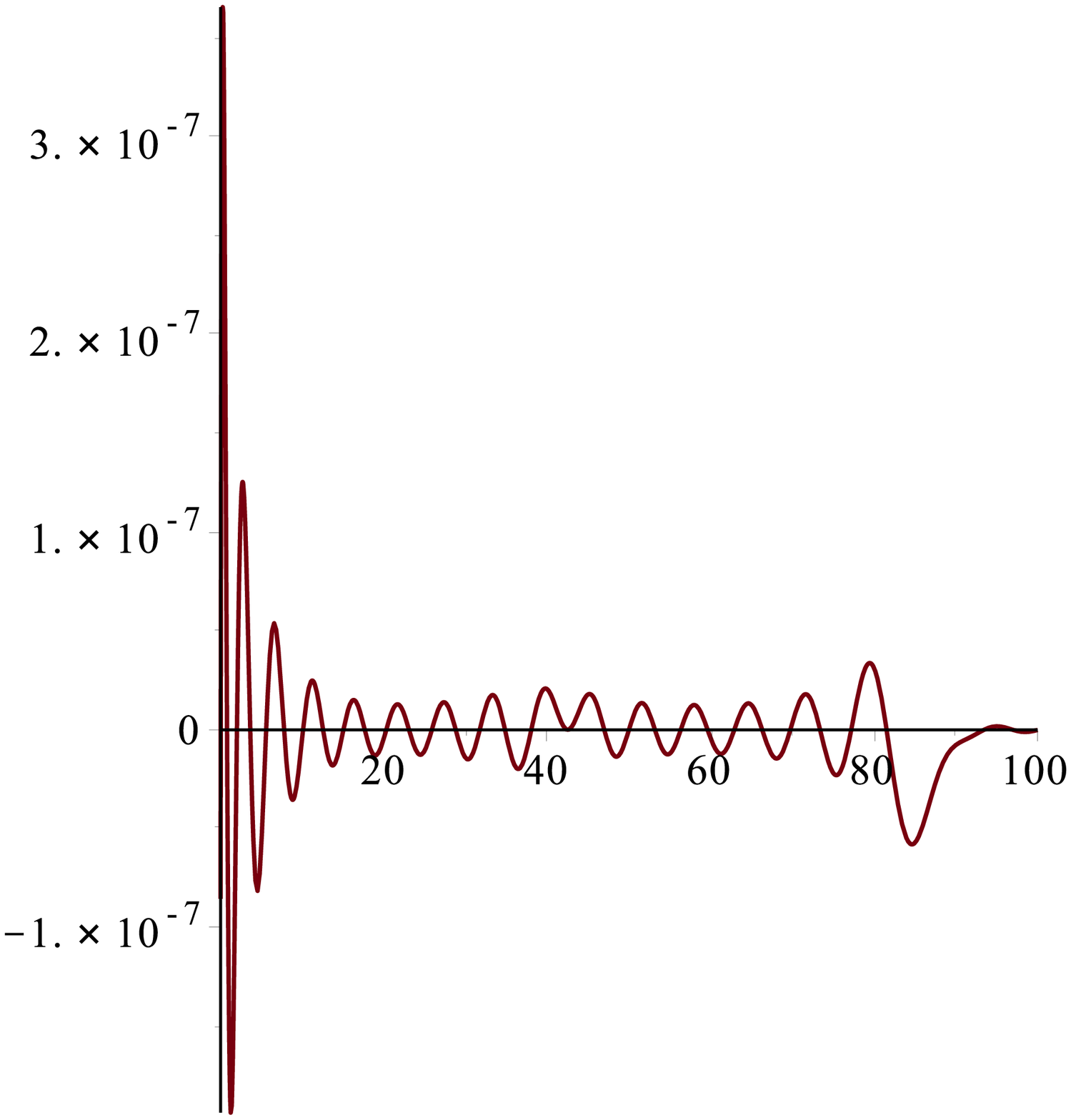}}
	\caption{Approximation of mortality density}
	\label{fig:app}
\end{figure}

\vspace{0.25cm}
{\noindent \bf (i) Survival model.}
Suppose that the variable annuity contract under consideration is issued to a $65$-year-old, whose survival model is determined by the Gompertz-Makeham law of mortality with the probability density given in \eqref{morden} where $x=65, A=0.0007, B=0.00005, c=10^{0.04}.$ Using the Hankel matrix method, we approximate the mortality density by a combination of $M=15$ terms of exponential functions. The bases and weights of the $15$-term exponential sum are shown in Figure \ref{fig:weight}. In Figure \ref{fig:app}, we show the plot of the original density function as well as the error from the $15$-term approximating exponential sum. It is clear from the plots that the maximum error is controlled, \[\sup_{t\in [0,100]} \left|f(t)-\sum^M_{i=1} w_i e^{-s_i t}\right|<10^{-6}.\]

\vspace{0.25cm}
{\noindent \bf (ii) Equity model.} Suppose that the variable annuity contract is invested in a single equity fund which is driven by either of the following two models
	\begin{enumerate}
		\item Geometric Brownian motion (GBM): Here we use a standard model from the insurance industry calibrated to monthly S\&P 500 total return data from December 1955 to December 2003 inclusive. The model is also known to pass the calibration criteria for equity return models set by the AAA (c.f. AAA report \cite[p.35]{AAA}).
 \[ \mu_1=0.064161,  \sigma_1=0.16.\] 
		\item Exponential L\'evy process with bilateral exponential jumps (Kou): we employ two sets of parameters for comparison with the GBM model.
 \begin{align*} &\mbox{(Parameter set A)}\qquad  \mu_2=0.119161,  \sigma_2=0.100499, \lambda=1, p=0.3, \rho=20, \hat{\rho}=10; \\
&\mbox{(Parameter set B)}\qquad \mu_2=0.064186,  \sigma_2=0.144395, \lambda=0.00005, p=0.3, \rho=0.1, \hat{\rho}=0.2.
\end{align*} The parameters are chosen so that the first two moments of $X_1$ are kept the same for both the GBM model and the  Kou model, i.e.
\begin{align*}
\mu_1&=\mu_2+\frac{\lambda p}{\rho}-\frac{\lambda(1-p)}{\hat{\rho}},\\
\sigma_1^2&=\sigma_2^2+\frac{2\lambda p}{\rho^2}+\frac{2\lambda(1-p)}{\hat{\rho}^2}.
	\end{align*} The first set of parameters leads to relatively frequent occurrence of small jumps, whereas the second set of parameters is chosen to exhibit relatively rare occurrence of large jumps. 
	\end{enumerate}

\vspace{0.25cm}
{\noindent \bf (iii)  Fee schedule.} The initial purchase payment is assumed to be $F_0=1$. The guarantee level starts off at $G_0=1$ and the yield rate on the insurer's assets backing up the GMDB liability is given by $r=0.02$. The mortality and expenses (M\&E) fee is charged at the rate of $m=0.01$ per dollar of the policyholder's investment account per time unit. The GMDB rider charge rate is assumed to be $35\%$ of the M\&E fee rate, i.e. $m_d=0.0035$.  

Recall that the GBM model is in fact a special case of the Kou model. Hence we shall first use tail probabilities of the GMDB net liability under the GBM model as benchmarks against which the accuracy of corresponding results under the Kou model can be tested. In Table \ref{tbl:comp}, the last row of tail probabilities are computed by formula \eqref{appx} where $\tilde{\p}(s,K)$ is determined by formulas in Remark \ref{remBM}.  The rest of the table are by formula \eqref{appx} where $\tilde{\p}(s,K)$ is determined by formulas in Corollary \ref{cor_CDF}. For the ease of direct comparison with the GBM model, we set for the Kou model
\[ \mu_2=0.064161,  \sigma_2=0.16, p=0.3, \rho=20, \hat{\rho}=10. \] As expected, Table \ref{tbl:comp} indicates that the tail probability of the GMDB net liability under the Kou model converges point-wise to the corresponding result under the GBM model, as the intensity rate $\lambda$ of jumps declines to zero.

\begin{table}[t] \centering \begin{tabular}{|r|r|r|r|r|r|}
  \hline
   & $\lambda=1$ & $\lambda=0.01$ & $\lambda=0.0001$  & $\lambda=0.000001$  & GBM ($\lambda=0$) \\ \hline \hline
  $\p(L>0.2)$ & $0.4794368114$   & $0.0954727742$& $0.0927572184$ & $0.0927302874$ & $0.0927300396$   \\ \hline
  $\p(L>0.4) $& $0.3313624187$  & $0.03327852158$& $0.03185715421$  & $0.03184312600$ & $0.03184298681$ \\\hline
  $\p(L>0.6)$ & $0.1787553560$  & $0.06201911742$ &$0.005797295345$ & $0.005793340382$ & $0.005793300500$   \\ \hline
\end{tabular}
\caption{Tail probabilities for the GMDB net liability} \label{tbl:comp}
\end{table}

\begin{table}[t] \centering \begin{tabular}{|c|c|c|c|}
  \hline
   & Analytic & Monte Carlo & Monte Carlo \\ 
 &  &  ($N=1,000$) & ($N=100,000$) \\ \hline \hline
  $\p(L>0.2)$ & $0.4794368114$   & $0.4787000000$ & $0.4796620000$    \\
 & & $(0.0154956700)$ & $(0.0015078343)$\\  \hline
 Time & $11.097$ & $68.422203$  & $7107.196853$ \\ \hline
  $\p(L>0.4) $& $0.3313624187$  & $0.3342000000$ & $0.3321305000$   \\
& & $(0.0143218640)$ & $(0.0013534030)$ \\ \hline
 Time & $10.912$ & $-$ & $-$ \\ \hline
  $\p(L>0.6)$ & $0.1787553560$  & $0.1780$ & $0.1794875000$  \\
 & & $(0.0105481353)$ & $(0.0011432358)$ \\ \hline
 Time & $10.463$ & $-$ & $-$ \\ \hline
\end{tabular}
\caption{Tail probabilities for the GMDB net liability with $\lambda=1$} \label{tbl:comp2}
\end{table}

We can also test the accuracy of results on tail probabilites of GMDB net liability against those resulting from a Monte Carlo method. 
Take the case of $\lambda=1$ for example in Table \ref{tbl:comp2}. For the Monte Carlo method,
 we first employ an acceptance-rejection method to generate policyholders' remaining lifetimes from 
 the Gompertiz-Makeham law of mortality in \eqref{morden}. In each experiment, we simulate
  $N$ sample paths of the equity index based on the exponential Levy model from the beginning 
  to policyholders' times of death. Under each
   sample path, we determine the GMDB net liability by the Riemman sum corresponding to \eqref{Lform} 
   with a step size of $0.01$. The GMDB payment is assumed to be payable at the end of the time step upon death. The tail probabilities $\p(L>0.2), \p(L>0.4), \p(L>0.6)$ are estimated respectively by  the number of sample paths under which the GMDB net liability surpasses the thresholds $0.2, 0.4, 0.6$, respectively, divided by the total number of sample paths $N$. In Table \ref{tbl:comp2}, we report tail probability results from both analytic formulas and estimates from Monte Carlo simulations. Computing time is reported in seconds. All algorithms based on the Monte Carlo method are implemented in Matlab (version 2016a) whereas results from analytic formulas are obtained in Maple (version 2016.1). In addition, each Monte Carlo result is the mean of estimates from 20 independent experiments and the corresponding sample standard deviation is quoted in brackets. Observe that Monte Carlo simulations are very time consuming to reach accuracy up to three decimal places. Therefore, it is worthwhile performing the above analysis to develop analytical formulas, as they are in general much more efficient and more accurate than Monte Carlo simulations.

\begin{figure}[t]
	\centering
	\subfigure[All positive liabilities ]{\includegraphics[width=6cm]{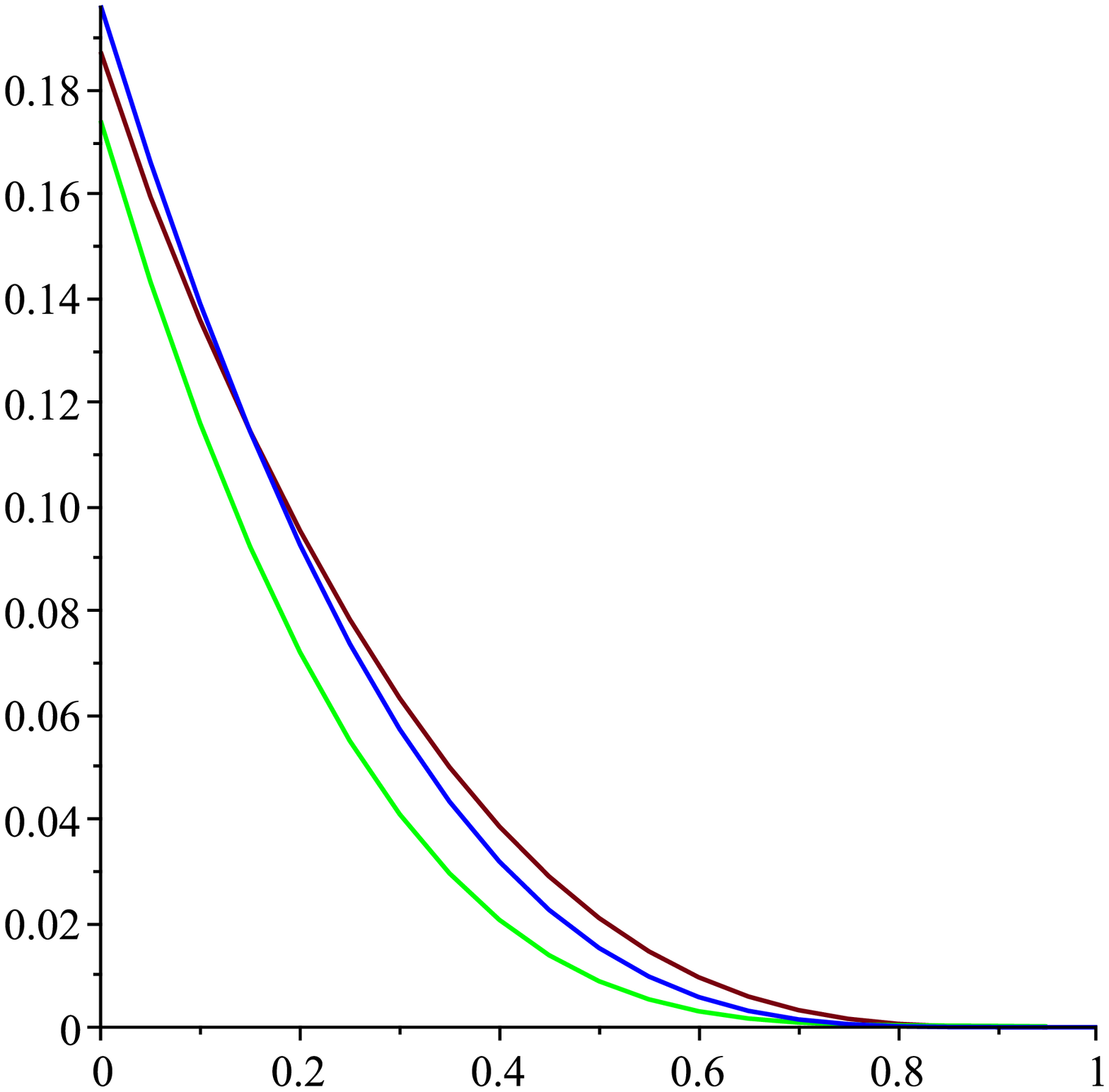}}\qquad
	\subfigure[Extremely large liabilities ]{\includegraphics[width=6cm]{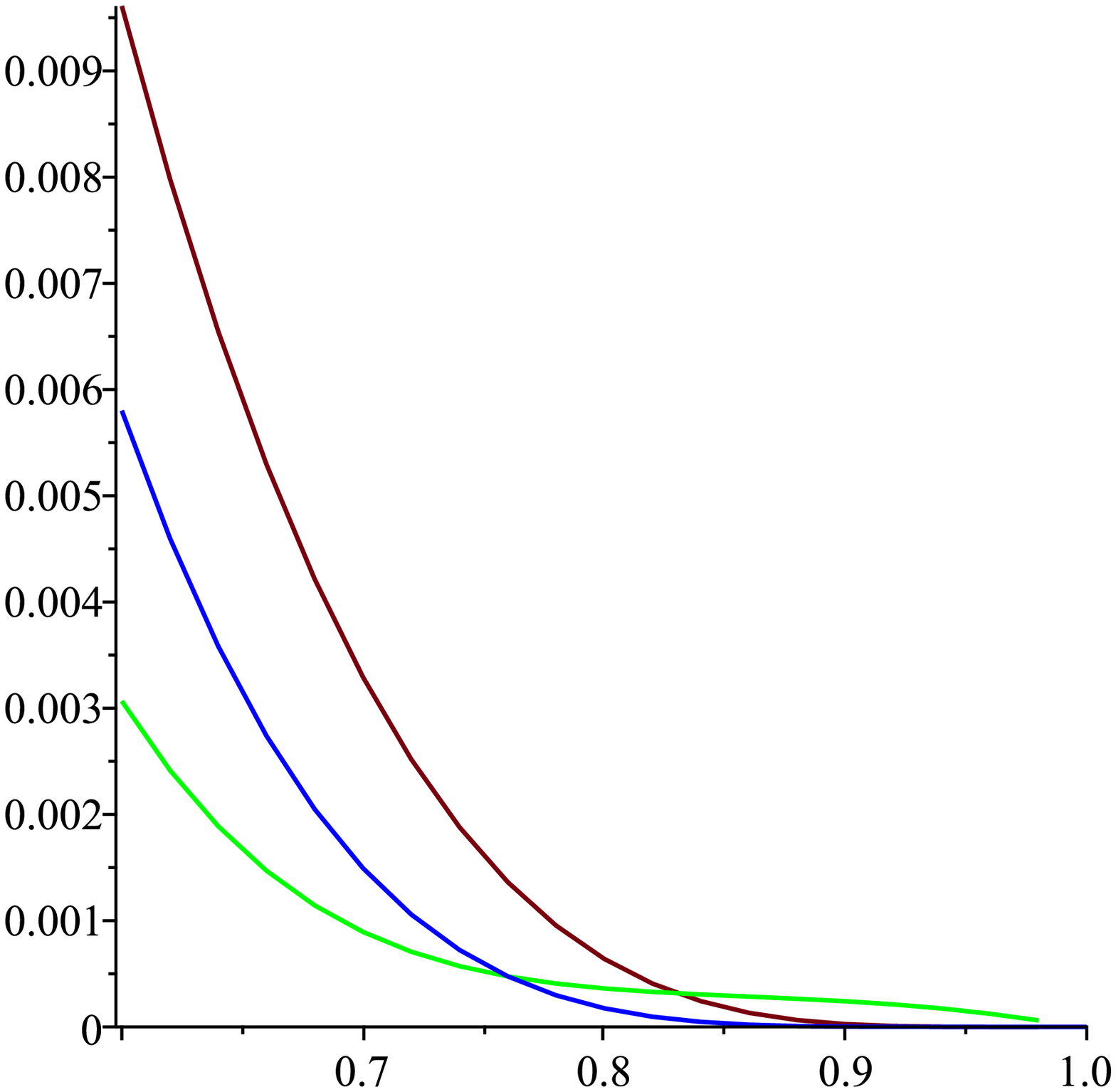}}
	\caption{Tail probability of GMDB net liability}
	\label{fig:comp}
\end{figure}

Owing to the analytical formulas developed in this paper, the computational algorithm for tail probability is very efficient, enabling us to plot the tail probability function. The visualization of tail probabilities allows us to develop an understanding of the impact of jumps to the overall riskiness of insurer's liability. For example, we plot tail probability functions of the GMDB rider under the GBM model and the Kou models. In Figure \ref{fig:comp}, the blue line represents the tail probability function under the GBM model whereas the red line and green line represent the tail probability function under the Kou models with parameter sets A and B respectively. The horizontal axis shows the level of net liability as a percentage of initial purchase payment and the vertical axis measures the corresponding tail probability. Figure \ref{fig:comp}(a) appears to indicate that the models with jumps tend to result in smaller probability of losses (positive net liability), which may be counterintuitive. This is likely caused by the fact that parameter sets A and B for the Kou models introduce smaller volatilities of white noise than that in the GBM model, which implies that larger probability masses are concentrated around negative net liabilities (profits for the insurer). The presence of jumps appears to play a role for generating extremely large liabilities, as shown in Figure \ref{fig:comp}(b). The tail probability in the Kou model with large jumps, represented by the green line, has a fatter tail than that in the GBM model, represented by blue line. The tail probability in the Kou model with smaller jumps, represented by the red line, also has a fatter tail, although to a less extent than the Kou model with large jumps. This is not surprising, as the equity index in Kou models with jumps can drop faster than the GBM can, thereby leading to severe losses for the insurer in extreme cases. This experiment shows that Kou models tend to produce more conservative estimates of insurer's net liabilities at the far right tail than the standard GBM model used in practice.

Next we illustrate the computation of risk measures for the GMDB net liability. The CTE$_{0.9}$ risk measure is commonly used to determine risk-based capitals for variable annuity guarantee products in the US. First we use the expression in \eqref{appx} to determine tail probability of GMDB net liability for various levels and then employ a bisection root search algorithm to determine the exact quantiles. The algorithm terminates when the search interval narrows down to a width less than $10^{-7}$. Then all results in Table \ref{tbl:riskmeas} are rounded to nearest sixth decimal place. Then the VaR results are fed into the algorithm for determining the CTE based on the expression \eqref{appxCTE}. Note that in Table \ref{tbl:riskmeas} both quantile and CTE risk measures at confidence levels $p=0.85, 0.9, 0.95$ for the model with parameter set A are larger than those in the model with parameter set B, which is consistent with the observation in Figure \ref{fig:comp}(a) that tail probability for the model with parameter set A (red line) tends to dominate that for the model with parameter set B (green line). However, if we move to the far right tail, the quantile and CTE risk measures at $p=0.9999$ for the model with parameter set A become less than those for the model with parameter set B, confirmed by the reversed dominance in Figure \ref{fig:comp}(b). Again the comparison of risk measures show that infrequent occurrence of large jumps only increases the tail probability at extremely high levels of liabilities whereas frequent occurrence of small jumps may significantly increase the tail probability at more modest levels of liabilities, which are often of interest to insurance applications.

\begin{table}[t] \centering \begin{tabular}{|r||r|r|r|r|}
  \hline
   & VaR$_{0.85}$ & VaR$_{0.9}$ & VaR$_{0.95}$ & VaR$_{0.9999}$  \\ \hline \hline
 Parameter set A& $0.069344$ &  $0.187615$  & $0.349984$& $0.868025$ \\ \hline
 Parameter set B & $0.038537$ & $0.132969$  & $0.266704$& $0.967712$  \\\hline \hline
    & CTE$_{0.85}$ & CTE$_{0.9}$ & CTE$_{0.95}$& CTE$_{0.9999}$ \\ \hline 
 Parameter set A& $0.295863$ & $0.380809$ & $0.498331$ &  $0.890319$ \\ \hline
 Parameter set B & $0.226736$  & $0.298245$ & $0.401757$  & $0.983389$ \\\hline
\end{tabular}
\caption{Risk measures for the GMDB net liability} \label{tbl:riskmeas}
\end{table}


\section*{Acknowledgements}

 The research of A. Kuznetsov was supported by the Natural Sciences and Engineering Research Council of Canada. 




\begin{appendices}

    \setcounter{proposition}{0}
    \renewcommand{\theproposition}{\Alph{section}\arabic{proposition}}
    \setcounter{theorem}{0}
    \renewcommand{\thetheorem}{\Alph{section}\arabic{theorem}}
    \setcounter{equation}{0}

\section{Meijer G-function and hypergeometric functions}\label{AppendixA}
    \renewcommand{\theequation}{\Alph{section}.\arabic{equation}}

    In this section we define Meijer G-functions and hypergeometric functions and discuss some of their properties. 
    We begin with four non-negative integers $m$, $n$, $p$ and $q$ and two vectors  
${\mathbf a}=(a_1,\dots,a_p) \in \c^p$ and ${\mathbf b}=(b_1,\dots,b_q) \in \c^q$ and define   for $0\le m \le q, 0 \le n \le p$,
\begin{equation}\label{product_Gammas}
{\mathcal{G}}^{mn}_{pq}\Big( 
\begin{matrix}
{\mathbf a} \\ {\mathbf b}
\end{matrix} \Big  \vert s 
\Big ) : =  \frac{\prod\limits_{j=1}^m \Gamma(b_j+s) \prod\limits_{j=1}^n \Gamma(1-a_j-s)}
{\prod\limits_{j=m+1}^q \Gamma(1-b_j-s) \prod\limits_{j=n+1}^p \Gamma(a_j+s)}. 
\end{equation}
We denote 
\begin{equation}\label{def_bar_a}
\underline{b}(m):=\min\limits_{1\le j \le m} \re(b_j), \;\;\; \bar{a}(n):=\max\limits_{1\le j \le n} \re(a_j), 
\end{equation}
and we set $\underline{b}(0)=+\infty$ and $\bar{a}(0)=-\infty$. When the parameters $m$, $n$, ${\mathbf a}$ and 
${\mathbf b}$ are fixed we will write simply $\underline{b}=\underline{b}(m)$ and 
$\bar{a}=\bar{a}(n)$.

\begin{definition}
Assume that parameters $m, n, p, q, {\bf a}$ and ${\bf b}$ satisfy the following two conditions
\begin{align}
\label{condition_A}
&{\textnormal{Condition A:}} \;\;\;\;\bar{a}-1<\underline{b}  \\
\label{condition_B}
&{\textnormal{Condition B:}}  \;\;\;\; p+q<2m+2n.
\end{align}
We define the {\it Meijer G-function} as follows
\begin{equation}\label{def_Meijer_G}
G^{mn}_{pq}\Big( 
\begin{matrix}
{\mathbf a} \\ {\mathbf b}
\end{matrix} \Big  \vert x 
\Big ):=\frac{1}{2\pi \i} \int_{\lambda + \i \r} 
{\mathcal{G}}^{mn}_{pq}\Big( 
\begin{matrix}
{\mathbf a} \\ {\mathbf b}
\end{matrix} \Big  \vert s 
\Big ) x^{-s} \d s,  
\end{equation}
where $x>0$ and $\lambda \in (-\underline{b}, 1-\bar{a})$. 
\end{definition}

Let us explain why the Meijer G-function is well-defined. The condition \eqref{condition_A} is needed because it separates the poles of $\Gamma(b_j+s)$ from the poles of $\Gamma(1-a_j-s)$ in the numerator in \eqref{product_Gammas}, thus the function 
$s\mapsto {\mathcal{G}}^{mn}_{pq}( {\mathbf a}, {\mathbf b}  | s)$ is analytic in the strip $-\underline b<\re(s)<1-\bar a$. Condition \eqref{condition_B}
and Stirling's asymptotic formula for the Gamma function ensure that the integrand in 
\eqref{def_Meijer_G} converges to zero exponentially fast as $\im(s) \to \infty$, and it is easy to check that \eqref{def_Meijer_G} defines the Meijer G-function as an analytic function in a sector $|\arg(z)|<(m+n-(p+q)/2) \pi$.

\begin{remark}
Our definition of  Meijer G-function is sufficient for our purposes, but it is not the most general possible. One could relax conditions \eqref{condition_A} and \eqref{condition_B} by appropriately deforming the contour of integration in 
\eqref{def_Meijer_G} . See Chapter 8.2 in Prudnikov et al. \cite{Prudnikov_V3} for more details.
\end{remark}

The hypergeometric function is defined as 
\begin{equation}\label{def_pFr}
{}_{p}F_{q} \Big (
\begin{matrix}
a_1, \dots, a_p \\
b_1, \dots, b_q
\end{matrix} \Big \vert 
z \Big):=\sum\limits_{k\ge 0} \frac{(a_1)_k \dots (a_p)_k}{(b_1)_k \dots (b_q)_k} \times \frac{z^k}{k!}, 
\end{equation}
where $(a)_k:=\Gamma(a+k)/\Gamma(a)$ is the Pochhammer symbol. We will also work with the regularized hypergeometric function 
\begin{equation}\label{def_pPhir}
{}_{p}\Phi_{q} \Big (
\begin{matrix}
a_1, \dots, a_p \\
b_1, \dots, b_q
\end{matrix} \Big \vert 
z \Big)= \Gamma\Big[ \begin{array}{c}
a_1, \dots, a_p \\
b_1, \dots, b_q
\end{array} \Big] {}_{p}F_{q} \Big (
\begin{matrix}
a_1, \dots, a_p \\
b_1, \dots, b_q
\end{matrix} \Big \vert 
z \Big).
\end{equation}

We record here some properties of Meijer G-function, which were used elsewhere in this paper. These properties and many other results on Meijer G-functions can be found in Gradshteyn and Ryzhik
\cite{Jeffrey2007}. In Chapter 8.4 in Prudnikov et al. \cite{Prudnikov_V3} one can find an extensive collection of formulas expressing various special functions in terms of Meijer G-functions.  
\begin{itemize}
\item[(i)] 
\begin{equation}\label{Meijer_G_times_x_power_c}
x^c G^{mn}_{pq}\Big( 
\begin{matrix}
{\mathbf a} \\ {\mathbf b}
\end{matrix} \Big  \vert x 
\Big )=G^{mn}_{pq}\Big( 
\begin{matrix}
{\mathbf a}+c \\ {\mathbf b}+c
\end{matrix} \Big  \vert x 
\Big ). 
\end{equation}
\item[(ii)] 
\begin{equation}\label{Meijer_G_x_to_1/x}
G^{mn}_{pq}\Big( 
\begin{matrix}
{\mathbf a} \\ {\mathbf b}
\end{matrix} \Big  \vert x 
\Big )=G^{nm}_{qp}\Big( 
\begin{matrix}
1-{\mathbf b} \\
1- {\mathbf a}
\end{matrix} \Big  \vert x^{-1} 
\Big ). 
\end{equation}

\item[(iii)]
For any $\epsilon>0$
\begin{equation}\label{Meijer_G_asymptotics}
 G^{mn}_{pq}\Big( 
\begin{matrix}
{\mathbf a} \\ {\mathbf b}
\end{matrix} \Big  \vert x 
\Big )= 
\begin{cases}
O(x^{\underline{b}-\epsilon}), \;\;\;\,\; {\textnormal{ as }} \; x\to 0^+,\\
O(x^{\bar{a}-1+\epsilon}), \; {\textnormal{ as }} \; x\to +\infty.
\end{cases}
\end{equation}
 \item[(iv)] 
Assume that $b_j-b_k \notin {\mathbb Z}$ for $1\le j<k \le m$. If $p<q$ or $p=q$ and $|x|<1$ we have 
\begin{align}\label{Meijer_G_in_terms_of_pFq}
G^{mn}_{pq}\Big( 
\begin{matrix}
{\mathbf a}  \\ {\mathbf b} 
\end{matrix} \Big  \vert x 
\Big )& =
\sum\limits_{k=1}^m 
\frac{\prod\limits_{\stackrel{1\le j \le m}{j\ne k}} \Gamma(b_j-b_k) \prod\limits_{j=1}^n \Gamma(1+b_k-a_j)}
{\prod\limits_{j=m+1}^q \Gamma(1+b_k-b_j) \prod\limits_{j=n+1}^p \Gamma(a_j-b_k)} \\
\nonumber
& \qquad \times 
x^{b_k}{}_pF_{q-1}
\Big( 
\begin{matrix}
1+b_k-a_1, \dots, 1+b_k-a_p \\ 1+b_k-b_1,\dots,*,\dots,1+b_k-b_q
\end{matrix} \Big  \vert (-1)^{p-m-n} x 
\Big ), 
\end{align}
where the asterisk in the function ${}_pF_{q-1}$ denotes the omission of the $k$-th parameter. If 
$p>q$ or $p=q$ and $|x|>1$, the corresponding representation of Meijer G-function in terms of  ${}_qF_{p-1}$ functions can be obtained using \eqref{Meijer_G_x_to_1/x} and \eqref{Meijer_G_in_terms_of_pFq}. 
\item[(v)]
If one of the parameter $a_j$ (for $j=1, 2, \cdots, n$) coincides with one of the parameters $b_j$ (for $j=m+1, m+2, \cdots, q$), the order of the G-function decreases. For example
\begin{equation}\label{Meijer_G_reduction}
G^{mn}_{pq}\Big( 
\begin{matrix}
a_1,\cdots,a_p \\ b_1,\cdots, b_{q-1},a_1
\end{matrix} \Big  \vert x 
\Big )=G^{m{, n-1}}_{{ p-1,q-1}}\Big( 
\begin{matrix}
a_2,\cdots,a_p \\ b_1,\cdots,b_{q-1}
\end{matrix} \Big  \vert x 
\Big ). 
\end{equation}
An analogous relationship occurs when one of the parameters $b_j$ (for $j=1, 2, \cdots, m$) coincides with one of the parameters $a_j$ (for $j=n+1, \cdots, p$). In this case, it is $m$ and not $n$ that decreases by one unit.
\begin{equation}\label{Meijer_G_reduction2}
G^{mn}_{pq}\Big( 
\begin{matrix}
a_1,\cdots,a_p \\ a_p, b_2\cdots, b_{q}
\end{matrix} \Big  \vert x 
\Big )=G^{m-1 n}_{ p-1,q-1}\Big( 
\begin{matrix}
a_1,\cdots,a_{p-1} \\ b_1,\cdots,b_{q-1}
\end{matrix} \Big  \vert x 
\Big ). 
\end{equation}
\item[(vi)]
\begin{equation}\label{Meijer_G_definite_integral}
\int_{1}^{\infty}x^{\alpha-1}G^{mn}_{pq}\Big( 
\begin{matrix}
{\mathbf a} \\ {\mathbf b}
\end{matrix} \Big  \vert zx 
\Big )\d x=G^{m+1,n}_{p+1,q+1}\Big( 
\begin{matrix}
{\mathbf a}, 1-\alpha \\ -\alpha, {\mathbf b}
\end{matrix} \Big  \vert z
\Big ). 
\end{equation}
\item[(vii)] {For $p\le q$ and $\mathrm{Re} (\alpha)>0$,}
\begin{equation}\label{pFq_definite_integral}
\int_{0}^1 {x}^{\alpha-1} 
{}_{p}F_{q} \Big (
\begin{matrix}
a_1, \dots, a_p \\
b_1, \dots, b_q
\end{matrix} \Big \vert 
zx \Big)\d x=\alpha^{-1}  \times  
{}_{p+1}F_{q+1} \Big (
\begin{matrix}
\alpha, a_1, \dots, a_p \\
\alpha+1, b_1, \dots, b_q
\end{matrix} \Big \vert 
z \Big). 
\end{equation}

\end{itemize}

\section{The proof of identity \eqref{identity_I1_I2_M}}\label{AppendixB}

We recall that $x>0$, $q>0$, the numbers 
$\{-\hat\zeta_2,-\hat \zeta_1, \zeta_1, \zeta_2\}$ and
$\{-\hat \rho,  \rho\}$
are the roots and the poles of the rational function $\psi(z)-q$ and they are known to satisfy 
the interlacing property
$$
-\hat \zeta_2<-\hat \rho<-\hat \zeta_1 < 0 < \zeta_1 <\rho <\zeta_2. 
$$ 
Note that the function $\psi(z)-q$ can be factorized as follows
\begin{equation}\label{psi_factorization}
\psi(z)-q=A \frac{(z-\zeta_1)(z-\zeta_2)(z+\hat \zeta_1)(z+\hat \zeta_2)}{(z-\rho)(z+\hat \rho)}. 
\end{equation}
where $A:=\sigma^2/2$. This fact (and the result $\psi(0)=0$) implies 
\begin{equation}
q=A\frac{\zeta_1 \zeta_2 \hat \zeta_1 \hat \zeta_2}{\rho \hat \rho},
\end{equation}
and 
\begin{align}
\psi'(\zeta_1)&=A \frac{(\zeta_1-\zeta_2)(\zeta_1+\hat \zeta_1)(\zeta_1+\hat \zeta_2)}{(\zeta_1-\rho)(\zeta_1+\hat \rho)}, \\
\psi'(\zeta_2)&=A \frac{(\zeta_2-\zeta_1)(\zeta_2+\hat \zeta_1)(\zeta_2+\hat \zeta_2)}{(\zeta_2-\rho)(\zeta_2+\hat \rho)}. 
\end{align}
Finally, we recall that we work under the following assumptions 
$$
\zeta_2-\zeta_1 \notin {\mathbb N}, \qquad \hat \zeta_2-\hat \zeta_1 \notin {\mathbb N}, \qquad 0 \vee (1-\hat \zeta_1) < \re(s)< 1.
$$ 
Our goal is to the identity \eqref{identity_I1_I2_M}, that we reproduce here for convenience:
\begin{equation}\label{identity_I1_I2_M2}
I_1(s)+I_2(s)-{\mathcal M}_{x,q}(s)=0.
\end{equation}
We remind that we denoted 
\begin{align}\label{def_I1}
\nonumber
I_1(s)&=\Big \{q A^{-\hat \zeta_1} x^{s-\hat \zeta_1}\frac{\sin(\pi(\hat \rho-\hat \zeta_1))}{\sin(\pi(\hat \zeta_2-\hat \zeta_1))}
{}_3\Phi_3 \Big( \begin{array}{c}
\hat \zeta_1, 1+\hat \zeta_1+\rho, 1+\hat \zeta_1-\hat \rho \\
1+\hat \zeta_1-\hat \zeta_2, 1+\hat \zeta_1+\zeta_1, 1+\hat \zeta_1+\zeta_2
\end{array} \Big \vert \frac{1}{Ax} \Big) 
 \\ 
&\qquad \qquad\qquad \qquad\qquad \;\;\; \times
G_{4,5}^{4,1}
\Big( \begin{array}{c}
 1-\hat \rho, 1, 1+\rho, s+1 \\
 s, 1+\zeta_1, 1+\zeta_2, 1-\hat \zeta_1, 1-\hat \zeta_2
\end{array} \Big \vert \frac{1}{Ax} \Big)\Big\}\\ \nonumber
\nonumber &+ \Big\{{\textnormal{the same expression with $\hat \zeta_1$ and  $\hat \zeta_2$ interchanged}}\Big\}, 
\end{align}
and
\begin{align}\label{def_I2}
I_2(s)&=\Big\{\frac{qx^{\zeta_1}+\zeta_1 {\mathcal M}_{x,q}(\zeta_1)}{\psi'(\zeta_1)}
\frac{x^{s-1-\zeta_1}}{1+\zeta_1-s} 
{}_4F_4 \Big( \begin{array}{c}
1+\zeta_1-s, 1+\zeta_1, 1+\zeta_1-\rho, 1+\zeta_1+\hat \rho \\ 
2+\zeta_1-s, 1+\zeta_1-\zeta_2, 1+\zeta_1+\hat \zeta_1, 1+\zeta_1+\hat \zeta_2 
\end{array} \Big \vert -\frac{1}{Ax} \Big)\Big\} \\  \nonumber
\nonumber
&+ \Big\{{\textnormal{the same expression with $\zeta_1$ and  $\zeta_2$ interchanged}}\Big\}, 
\end{align}
and we have computed earlier
\begin{align}\label{eqn_formula_Mxq2}
{\mathcal M}_{x,q}(s)&=q A^{-s} 
\Gamma\Big[ \begin{array}{c}
1+\zeta_1-s, \; 1+\zeta_2-s, \; \hat \rho+s \\
1-s, \; 1+\rho-s, \;\hat \zeta_1+s, \;\hat \zeta_2+s 
\end{array} \Big]\\
&\times 
G_{4,5}^{3,3}
\Big( \begin{array}{c}
1-s, 1, -\rho, \hat \rho \\
1-s, \hat \zeta_1, \hat \zeta_2, -\zeta_1,-\zeta_2 
\end{array} \Big \vert \frac{1}{Ax} \Big).
\nonumber
\end{align}

Our main tool will be formula \eqref{Meijer_G_in_terms_of_pFq}, which expresses Meijer G-function as a sum of hypergeometric functions. 
The following variant of \eqref{Meijer_G_in_terms_of_pFq} also will be used frequently: 
\begin{align}\label{Meijer_G_in_terms_of_pPhiq}
G^{mn}_{pq}\Big( 
\begin{matrix}
{\mathbf a}  \\ {\mathbf b} 
\end{matrix} \Big  \vert x 
\Big )& = \pi^{m+n-p-1}
\sum\limits_{k=1}^m 
\frac{\prod\limits_{j=n+1}^p \sin(\pi (a_j-b_k))}
{\prod\limits_{\stackrel{1\le j \le m}{j\ne k}} \sin(\pi(b_j-b_k))} \\
\nonumber
& \qquad \times 
x^{b_k}{}_p\Phi_{q-1}
\Big( 
\begin{matrix}
1+b_k-a_1, \dots, 1+b_k-a_p \\ 1+b_k-b_1,\dots,*,\dots,1+b_k-b_q
\end{matrix} \Big  \vert (-1)^{p-m-n} x 
\Big ), 
\end{align}
where we assume that $b_j-b_k \notin {\mathbb Z}$ for $1\le j<k \le m$ and $p<q$. 
This formula can be easily derived from \eqref{Meijer_G_in_terms_of_pFq} using {\it the reflection formula} for the Gamma function:
\begin{equation}\label{gamma_reflection}
\Gamma(z)\Gamma(1-z) = \frac{\pi}{\sin(\pi z)}. 
\end{equation}

\vspace{0.25cm}
\noindent
{\bf Proof  of  the identity \eqref{identity_I1_I2_M2}.}
As the proof will be rather technical and will involve many tedious computations, let us explain the main steps and ideas behind the proof. 
The first step is to express all Meijer G-functions appearing in \eqref{def_I1}, \eqref{def_I2} and \eqref{eqn_formula_Mxq2} in terms of hypergeometric functions
via \eqref{Meijer_G_in_terms_of_pFq} or \eqref{Meijer_G_in_terms_of_pPhiq}. In the second step we will use the results of step one and we will rewrite the expression in \eqref{identity_I1_I2_M2} as a sum of products of two hypergeometric functions. 
In the third step our goal is to simplify the expression obtained in step two. In the fourth step we will show that the resulting (simplified) identity is true because it is a special case of a more general result \cite[Theorem 1]{FKY2015}. 

Let us deal with the first step -- expressing Meijer G-functions in terms of hypergeometric functions.  

\vspace{0.25cm}
\noindent
{\bf Step 1a. }
We define 
\begin{align*}
f_1&:={}_4\Phi _4 \Big( \begin{array}{c}
1, 1-s, 2+\rho-s, 2-\hat \rho-s \\
2-s-\hat \zeta_1, 2-s-\hat \zeta_2, 2-s+\zeta_1, 2-s+\zeta_2
\end{array} \Big \vert \frac{1}{Ax} \Big), \\
f_2&:={}_3\Phi _3 \Big( \begin{array}{c}
\hat\zeta_1,1+\hat\zeta_1+\rho,1+\hat\zeta_1-\hat\rho \\ 1+\hat\zeta_1-\hat\zeta_2,1+\hat\zeta_1+\zeta_1,1+\hat\zeta_1+\zeta_2
\end{array} \Big \vert \frac{1}{Ax} \Big),\\
f_3&:={}_3\Phi _3 \Big( \begin{array}{c}
\hat\zeta_2,1+\hat\zeta_2+\rho,1+\hat\zeta_2-\hat\rho \\ 
1+\hat\zeta_2-\hat\zeta_1,1+\hat\zeta_2+\zeta_1,1+\hat\zeta_2+\zeta_2
\end{array} \Big \vert \frac{1}{Ax} \Big),
\end{align*}
and 
\begin{align*}
a_1:=-\pi q A^{-s} \Gamma
\Big[ \begin{array}{c}
1+\zeta_1-s, 1+\zeta_2-s, \hat\rho+s \\ 
1-s, 1+\rho-s, \hat\zeta_1+s, \hat\zeta_2+s
\end{array} \Big] 
\frac{\sin(\pi(\hat\rho+s))}{\sin(\pi (\hat\zeta_1+s))\sin(\pi(\hat\zeta_2+s))} (Ax)^{s-1},\\
a_2:=\pi q A^{-s} \Gamma
\Big[ \begin{array}{c}
1+\zeta_1-s, 1+\zeta_2-s, \hat\rho+s \\ 
1-s, 1+\rho-s, \hat\zeta_1+s, \hat\zeta_2+s
\end{array} \Big] 
\frac{\sin(\pi(\hat\rho-\hat\zeta_1))}{\sin(\pi (s+\hat\zeta_1))\sin(\pi(\hat\zeta_2-\hat\zeta_1))} (Ax)^{-\hat\zeta_1}, \\
a_3:=\pi q A^{-s} \Gamma
\Big[ \begin{array}{c}
1+\zeta_1-s, 1+\zeta_2-s, \hat\rho+s \\ 
1-s, 1+\rho-s, \hat\zeta_1+s, \hat\zeta_2+s
\end{array} \Big] 
\frac{\sin(\pi(\hat\rho-\hat\zeta_2))}{\sin(\pi (s+\hat\zeta_2))\sin(\pi(\hat\zeta_1-\hat\zeta_2))} (Ax)^{-\hat\zeta_2}. 
\end{align*}
Then formulas \eqref{eqn_formula_Mxq2} and \eqref{Meijer_G_in_terms_of_pPhiq} give us 
\begin{equation}\label{formula1}
{\mathcal M}_{x,q}(s)=a_1 f_1+a_2 f_2+a_3f_3. 
\end{equation}

\vspace{0.25cm}
\noindent
{\bf Step 1b. }
We define $f_4:=f_1 \big \vert_{s=\zeta_1}$ and $f_5:=f_1 \big \vert_{s=\zeta_2}$, that is 
\begin{align*}
f_4&:={}_4\Phi _4 \Big( \begin{array}{c}
1,1-\zeta_1,2+\rho-\zeta_1,2-\hat\rho-\zeta_1 \\
2,2-\zeta_1-\hat\zeta_1,2-\zeta_1-\hat\zeta_2,2-\zeta_1+\zeta_2
\end{array} \Big \vert \frac{1}{Ax} \Big)\\
f_5&:={}_4\Phi _4 \Big( \begin{array}{c}
1,1-\zeta_2,2+\rho-\zeta_2,2-\hat\rho-\zeta_2 \\ 2,2-\zeta_2-\hat\zeta_1,2-\zeta_2-\hat\zeta_2,2-\zeta_2+\zeta_1
\end{array} \Big \vert \frac{1}{Ax} \Big).
\end{align*}
In the same way we define
\begin{align*}
b_1:=a_1 \big \vert_{s=\zeta_1}, \qquad b_2:=a_2 \big \vert_{s=\zeta_1}, 
\qquad b_3:=a_3 \big \vert_{s=\zeta_1},  \\
c_1:=a_1 \big \vert_{s=\zeta_2}, \qquad c_2:=a_2 \big \vert_{s=\zeta_2}, 
\qquad c_3:=a_3 \big \vert_{s=\zeta_2}.  
\end{align*}
Then \eqref{formula1} gives us
\begin{align}\label{formula2}
{\mathcal M}_{x,q}(\zeta_1)&=b_1f_4+b_2f_2+b_3f_3,\\ \nonumber
{\mathcal M}_{x,q}(\zeta_2)&=c_1f_5+c_2f_2+c_3f_3. 
\end{align}

\vspace{0.25cm}
\noindent
{\bf Step 1c. }
We define 
\begin{align*}
f_6&:={}_4\Phi _4 \Big( \begin{array}{c}
1+\zeta_1+\hat\rho,1+\zeta_1,1+\zeta_1-\rho,1+\zeta_1-s,\\ 
2+\zeta_1-s,1+\zeta_1-\zeta_2,1+\zeta_1+\hat\zeta_1,1+\zeta_1+\hat\zeta_2
\end{array} \Big \vert -\frac{1}{Ax} \Big), \\
f_7&:={}_4\Phi _4 \Big( \begin{array}{c}
1+\zeta_2+\hat\rho,1+\zeta_2,1+\zeta_2-\rho,1+\zeta_2-s \\
2+\zeta_2-s,1+\zeta_2-\zeta_1,1+\zeta_2+\hat\zeta_1,1+\zeta_2+\hat\zeta_2
\end{array} \Big \vert -\frac{1}{Ax} \Big),\\
f_8&:={}_4\Phi _4 \Big( \begin{array}{c}
1+\hat\rho-\hat\zeta_1,1-\hat\zeta_1,1-\rho-\hat\zeta_1,1-s-\hat\zeta_1 \\
2-\hat\zeta_1-s,1-\hat\zeta_1-\zeta_1,1-\hat\zeta_1-\zeta_2,1+\hat\zeta_2-\hat\zeta_1
\end{array} \Big \vert -\frac{1}{Ax} \Big),
\end{align*}
and 
\begin{align*}
&d_1:=-
\frac{\sin(\pi \zeta_1)\sin(\pi(\rho- \zeta_1))}{\sin(\pi (\zeta_2-\zeta_1))\sin(\pi(\hat\zeta_1+\zeta_1))} (Ax)^{-\zeta_1-1},\\
&d_2:=-
\frac{\sin(\pi \zeta_2)\sin(\pi(\rho- \zeta_2))}{\sin(\pi (\zeta_1-\zeta_2))\sin(\pi(\hat\zeta_1+\zeta_2))} (Ax)^{-\zeta_2-1},\\
&d_3:=-
\frac{\sin(\pi \hat \zeta_1)\sin(\pi(\rho+ \hat\zeta_1))}{\sin(\pi (\zeta_1+\hat \zeta_1))\sin(\pi(\zeta_2+\hat\zeta_1))} (Ax)^{\hat\zeta_1-1},\\
&d_4:=
\Gamma
\Big[ \begin{array}{c}
1+\zeta_1-s,1+\zeta_2-s,1-\hat\zeta_1-s,s+\hat\rho \\ 
s+\hat\zeta_2, 1-s, 1+\rho-s
\end{array} \Big] (Ax)^{-s}.
\end{align*} 
Then formulas \eqref{Meijer_G_in_terms_of_pFq} and \eqref{Meijer_G_in_terms_of_pPhiq} give us 
\begin{equation}\label{formula3}
G_{4,5}^{4,1}
\Big( \begin{array}{c}
 1-\hat \rho, 1, 1+\rho, s+1 \\
 s, 1+\zeta_1, 1+\zeta_2, 1-\hat \zeta_1, 1-\hat \zeta_2
\end{array} \Big \vert \frac{1}{Ax} \Big)=d_1 f_6+d_2 f_7+d_3 f_8+d_4. 
\end{equation} 

\vspace{0.25cm}
\noindent
{\bf Step 1d. }
We define 
$$
f_9:={}_4\Phi _4 \Big( \begin{array}{c}
1+\hat\rho-\hat\zeta_2,1-\hat\zeta_2,1-\rho-\hat\zeta_2,1-s-\hat\zeta_2 \\
2-\hat\zeta_2-s,1-\hat\zeta_2-\zeta_1,1-\hat\zeta_2-\zeta_2,1+\hat\zeta_1-\hat\zeta_2
\end{array} \Big \vert -\frac{1}{Ax} \Big),
$$
and 
\begin{align*}
&e_1:=-
\frac{\sin(\pi \zeta_1)\sin(\pi(\rho- \zeta_1))}{\sin(\pi (\zeta_2-\zeta_1))\sin(\pi(\hat\zeta_2+\zeta_1))} (Ax)^{-\zeta_1-1},\\
&e_2:=-
\frac{\sin(\pi \zeta_2)\sin(\pi(\rho- \zeta_2))}{\sin(\pi (\zeta_1-\zeta_2))\sin(\pi(\hat\zeta_2+\zeta_2))} (Ax)^{-\zeta_2-1},\\
&e_3:=-
\frac{\sin(\pi \hat \zeta_2)\sin(\pi(\rho+ \hat\zeta_2))}{\sin(\pi (\zeta_1+\hat \zeta_2))\sin(\pi(\zeta_2+\hat\zeta_2))} (Ax)^{\hat\zeta_2-1},\\
&e_4:=
\Gamma
\Big[ \begin{array}{c}
1+\zeta_1-s,1+\zeta_2-s,1-\hat\zeta_2-s,s+\hat\rho \\ 
s+\hat\zeta_1, 1-s, 1+\rho-s
\end{array} \Big] (Ax)^{-s}.
\end{align*} 
Then formulas \eqref{Meijer_G_in_terms_of_pFq} and \eqref{Meijer_G_in_terms_of_pPhiq} give us 
\begin{equation}\label{formula4}
G_{4,5}^{4,1}
\Big( \begin{array}{c}
 1-\hat \rho, 1, 1+\rho, s+1 \\
 s,1+\zeta_1,1+\zeta_2,1-\hat\zeta_2,1-\hat\zeta_1
\end{array} \Big \vert \frac{1}{Ax} \Big)=e_1f_6+e_2f_7+e_3f_9+e_4. 
\end{equation}

\vspace{0.25cm}
Our next goal is to collect all these formulas and express the functions $I_1(s)$ and $I_2(s)$ as sums of products $f_if_j$. 

\vspace{0.25cm}
\noindent
{\bf Step 2a. }
We define 
\begin{align*}
h_1=q A^{-\hat\zeta_1} x^{s-\hat\zeta_1}
\frac{\sin(\pi(\hat\rho-\hat\zeta_1))}{\sin(\pi(\hat\zeta_2-\hat\zeta_1))},\\
h_2=q A^{-\hat\zeta_2} x^{s-\hat\zeta_2}
\frac{\sin(\pi(\hat\rho-\hat\zeta_2))}{\sin(\pi(\hat\zeta_1-\hat\zeta_2))},
\end{align*}
and from formulas \eqref{def_I1}, \eqref{formula3} and \eqref{formula4} we obtain
\begin{align}\label{formula6}
I_1(s)&=(h_1d_4)f_2+(h_2e_4)f_3+(h_1d_1)f_2f_6+(h_1d_2)f_2f_7
\\ \nonumber
&+(h_1d_3)f_2f_8+(h_2e_1)f_3f_6+(h_2e_2)f_3f_7+(h_2e_3)f_3f_9.
\end{align}

\vspace{0.25cm}
\noindent
{\bf Step 2b. }
We define 
\begin{align*}
g_1:&=\frac{1}{\psi'(\zeta_1)} \frac{x^{s-1-\zeta_1}}{1+\zeta_1-s} 
\Gamma
\Big[ \begin{array}{c}
2+\zeta_1-s,1+\zeta_1-\zeta_2,1+\zeta_1+\hat\zeta_1,1+\zeta_1+\hat\zeta_2 \\ 
1+\zeta_1+\hat\rho,1+\zeta_1,1+\zeta_1-\rho,1+\zeta_1-s
\end{array} \Big], \\
g_2:&=\frac{1}{\psi'(\zeta_2)} \frac{x^{s-1-\zeta_2}}{1+\zeta_2-s} 
\Gamma
\Big[ \begin{array}{c}
2+\zeta_2-s,1+\zeta_2-\zeta_1,1+\zeta_2+\hat\zeta_1,1+\zeta_2+\hat\zeta_2\\ 
1+\zeta_2+\hat\rho,1+\zeta_2,1+\zeta_2-\rho,1+\zeta_2-s
\end{array} \Big],
\end{align*}
and from formulas  \eqref{def_I2} and \eqref{formula2} we obtain
\begin{align}\label{formula5}
I_2(s)&=(q x^{\zeta_1} g_1)f_6+(qx^{\zeta_2}g_2)f_7+(\zeta_1 b_1g_1)f_4f_6+(\zeta_1b_2g_1)f_2f_6\\ 
\nonumber
&+(\zeta_1b_3g_1)f_3f_6+(\zeta_2c_1g_2)f_5f_7+(\zeta_2c_2g_2)f_2f_7+(\zeta_2c_3g_2)f_3f_7.
\end{align}

\vspace{0.25cm}
\noindent
{\bf Step 2c. } Using all the previous results (formulas \eqref{formula1}, \eqref{formula6} and \eqref{formula5}) we rewrite the identity \eqref{identity_I1_I2_M2} in an equivalent form
\begin{align}\label{formula7}
\nonumber
I_1(s)+I_2(s)-{\mathcal M}_{x,q}(s)&=
(h_1d_4)f_2+(h_2e_4)f_3+(h_1d_1)f_2f_6+(h_1d_2)f_2f_7
\\ \nonumber
&+(h_1d_3)f_2f_8+(h_2e_1)f_3f_6+(h_2e_2)f_3f_7+(h_2e_3)f_3f_9\\
&+(q x^{\zeta_1} g_1)f_6+(qx^{\zeta_2}g_2)f_7+(\zeta_1 b_1g_1)f_4f_6+(\zeta_1b_2g_1)f_2f_6\\ 
\nonumber
&+(\zeta_1b_3g_1)f_3f_6+(\zeta_2c_1g_2)f_5f_7+(\zeta_2c_2g_2)f_2f_7+(\zeta_2c_3g_2)f_3f_7\\ \nonumber
&-(a_1 f_1+a_2 f_2+a_3f_3)=0. 
\end{align}

\vspace{0.25cm}
Now our goal is  to simplify the long sum in \eqref{formula7}. First we will deal with cancellations and then we will use a certain transformation of hypergeometric functions.

\vspace{0.25cm}
\noindent
{\bf Step 3a. } Using  the reflection formula for the Gamma function \eqref{gamma_reflection} we check that 
\begin{align*}
&\zeta_1b_3g_1=-h_2e_1,\\ 
&\zeta_1b_2g_1=-h_1d_1,\\
&\zeta_2c_2g_2=-h_1d_2,\\
&\zeta_2c_3g_2=-h_2e_2,\\
&a_2=h_1d_4,\\
&a_3=h_2e_4.
\end{align*}
These identities allow us to simplify the expression in \eqref{formula7} as follows
\begin{align}\label{formula8}
\nonumber
I_1(s)+I_2(s)-{\mathcal M}_{x,q}(s)&=
(qx^{\zeta_1}+\zeta_1b_1f_4)g_1f_6+(qx^{\zeta_2}+\zeta_2c_1f_5)g_2f_7\\ 
&+(h_1d_3)f_2f_8+(h_2e_3)f_3f_9-a_1f_1. 
\end{align}

\vspace{0.25cm}
\noindent
{\bf Step 3b. }
In this step we will simplify \eqref{formula8} via the following result
\begin{equation}\label{eqn_4F4_3F3}
1+z\Big(\prod_{i=1}^3 \frac{\alpha_i-1}{\beta_i-1}\Big) \times {}_4F _4 \Big( \begin{array}{c}
1, \alpha_1, \alpha_2, \alpha_3 \\
2, \beta_1, \beta_2, \beta_3
\end{array} \Big \vert z \Big)=
{}_3F _3 \Big( \begin{array}{c}
\alpha_1-1, \alpha_2-1, \alpha_3-1 \\
\beta_1-1, \beta_2-1, \beta_3-1
\end{array} \Big \vert z \Big). 
\end{equation}
The above identity can be easily established by comparing the coefficients of the Taylor series of both sides.  Applying identity \eqref{eqn_4F4_3F3} we obtain
\begin{align}
\label{formula9}
qx^{\zeta_1}+\zeta_1b_1f_4&=qx^{\zeta_1}f_{10},\\
\label{formula10}
qx^{\zeta_2}+\zeta_2c_1f_5&=qx^{\zeta_2}f_{11}. 
\end{align}
where 
\begin{align*}
f_{10}&:={}_3F _3 \Big( \begin{array}{c}
-\zeta_1,1+\rho-\zeta_1,1-\hat\rho-\zeta_1 \\
1-\zeta_1-\hat\zeta_1,1-\zeta_1-\hat\zeta_2,1-\zeta_1+\zeta_2
\end{array} \Big \vert \frac{1}{Ax} \Big),\\
f_{11}&:={}_3F _3 \Big( \begin{array}{c}
-\zeta_2,1+\rho-\zeta_2,1-\hat\rho-\zeta_2 \\
1-\zeta_2-\hat\zeta_1,1-\zeta_2-\hat\zeta_2,1-\zeta_2+\zeta_1
\end{array} \Big \vert \frac{1}{Ax} \Big).
\end{align*}
Formulas \eqref{formula8}, \eqref{formula9} and \eqref{formula10} give us an equivalent form of the identity 
$I_1(s)+I_2(s)-{\mathcal M}_{x,q}(s)=0$ as follows
\begin{align}\label{formula11}
(qx^{\zeta_1}g_1)f_6f_{10}
+(qx^{\zeta_2}g_2)f_7f_{11}+(h_1d_3)f_2f_8+(h_2e_3)f_3f_9-a_1f_1=0. 
\end{align}
Now it remains to prove \eqref{formula11}.

\vspace{0.25cm}
\noindent
{\bf Step 4. }
 By simplifying the coefficients 
(again, using the reflection formula for the Gamma function \eqref{gamma_reflection}) one can check that the left-hand side in \eqref{formula11} is a finite (that is, non-infinite) multiple of 
\begin{align*}
H(x):=\sum\limits_{i=1}^{5} 
\frac{(\alpha_i-\rho)(\alpha_i+\hat \rho)}{ \prod\limits_{\stackrel{1\le j \le 5}{j \neq i}} (\alpha_i-\alpha_j)}  &\times 
{}_{4}F_{4} \Big (
\begin{matrix}
1+\alpha_i-\rho, 1+\alpha_i+\hat \rho, 1+\alpha_i, 1+\alpha_i-s \\
1+\alpha_i-\alpha_1, \dots,*,\dots, 1+\alpha_i-\alpha_{5}
\end{matrix} \Big \vert 
-\frac{1}{Ax} \Big)
\\ \nonumber
&\times 
{}_{4}F_{4} \Big (
\begin{matrix}
 1+\rho-\alpha_i,1-\hat \rho-\alpha_i, -\alpha_i, s-\alpha_i \\
1+\alpha_1-\alpha_i, \dots, *, \dots, 1+\alpha_{5}-\alpha_i
\end{matrix} \Big \vert 
\frac{1}{Ax} \Big), 
\end{align*}
where $[\alpha_1, \alpha_2, \alpha_3, \alpha_4, \alpha_5]=[\zeta_1, \zeta_2, -\hat \zeta_1, -\hat \zeta_2, s-1]$
and the asterisk means that the term $1+\alpha_i-\alpha_i$ is omitted. The identity $H(x)\equiv 0$ is a special case of  \cite[Theorem 1]{FKY2015}. To see this,  we should set $p=r=4$ and
\begin{align*}
\{a_i\}_{1\le i \le 5}=\{\zeta_1,\zeta_2,-\hat \zeta_1, -\hat \zeta_2, s-1\}, \qquad
\{b_i\}_{1\le i \le 4}=\{1+\rho, 1-\hat \rho, 0, s\}, \qquad
\{m_i\}_{1\le i \le 4}=\{1,1,0,0\}, 
\end{align*}  
in the notation of \cite[Theorem 1]{FKY2015}.  
\qed

\end{appendices}

\end{document}